\documentclass[english, a4paper]{article}

\usepackage[margin=1.65in]{geometry}
\usepackage[font=small, labelfont=bf, margin=0.35in]{caption}

\usepackage{titling}
\usepackage{authblk}

\usepackage{babel}
\usepackage[utf8]{inputenc}
\usepackage[pdftex]{hyperref}

\usepackage{color}

\usepackage{amsmath}
\usepackage{amssymb}
\usepackage{amsthm}

\usepackage{cite}

\usepackage{enumitem}

\usepackage{tikz}

\usepackage{multirow}
\usepackage{tabularx}

\makeatletter
\newcommand{\thickhline}{%
    \noalign {\ifnum 0=`}\fi \hrule height 1pt%
    \futurelet \reserved@a \@xhline%
}
\newcolumntype{x}{@{\hskip\tabcolsep\vrule width 1pt\hskip\tabcolsep}}
\makeatother

\usepackage[labelfont=bf, figurename=Fig., tablename=Tab.]{caption}

\usepackage{algorithm}

\makeatletter
\newenvironment{breakablealgorithm}
  {
    \begin{center}
      \refstepcounter{algorithm}
      \hrule height.8pt depth0pt \kern2pt
      \renewcommand{\caption}[2][\relax]{
        {\raggedright\textbf{\noindent\fname@algorithm~\thealgorithm} ##2\par}%
        \ifx\relax##1\relax
          \addcontentsline{loa}{algorithm}{\protect\numberline{\thealgorithm}##2}%
        \else
          \addcontentsline{loa}{algorithm}{\protect\numberline{\thealgorithm}##1}%
        \fi
        \kern2pt\hrule\kern2pt
      }
  }{
      \kern2pt\hrule\relax
    \end{center}
  }
\makeatother

\newlist{pseudocode}{enumerate}{6}

\setlist[pseudocode]{ 
  label*={\arabic*.},
  ref={\arabic*},
  before=\raggedright,
  leftmargin=*,
  itemsep=1.5pt,
  before=\setlength{\listparindent}{0pt},
  parsep=1pt,
  topsep=4pt}

\setlist[pseudocode, 2, 3, 4, 5, 6]{
  topsep=0pt
}

\setlist[pseudocode, 2]{
  ref={\arabic{pseudocodei}.\arabic*},
}

\setlist[pseudocode, 3]{
  ref={\arabic{pseudocodei}.\arabic{pseudocodeii}.\arabic*}}

\setlist[pseudocode, 4]{
  ref={\arabic{pseudocodei}.\arabic{pseudocodeii}.\arabic{pseudocodeiii}.\arabic*}
}

\setlist[pseudocode, 5]{
  ref={\arabic{pseudocodei}.\arabic{pseudocodeii}.\arabic{pseudocodeiii}.\arabic{pseudocodeiv}.\arabic*},
}

\setlist[pseudocode, 6]{
  ref={\arabic{pseudocodei}.\arabic{pseudocodeii}.\arabic{pseudocodeiii}.\arabic{pseudocodeiv}.\arabic{pseudocodev}.\arabic*},
}

\allowdisplaybreaks

\newcommand{\ket}[1]{\left|\, #1 \, \right\rangle}
\newcommand{\e}{\mathrm{e}}
\newcommand{\imag}{\mathrm{i}}
\newcommand{\ceil}[1]{\left\lceil #1 \right\rceil}
\newcommand{\floor}[1]{\left\lfloor #1 \right\rfloor}
\newcommand{\round}[1]{\left\lfloor #1 \right\rceil}
\renewcommand{\vec}{\mathbf}
\renewcommand{\parallel}{=}
\newcommand{\poly}{\mathrm{poly}}

\newcommand{\Nmodulus}{M}
\newcommand{\Nspace}{N}
\newcommand{\nruns}{n}

\newcommand{\refeq}[1]{(\ref{eq:#1})}

\newtheorem{theorem}{Theorem}
\newtheorem{lemma}{Lemma}
\newtheorem{corollary}{Corollary}

\newtheorem{claim}{Claim}
\newtheorem{definition}{Definition}

\title{On the success probability of the quantum algorithm for the short DLP}
\author[1,2]{\href{mailto:ekera@kth.se}{Martin Ekerå}}
\affil[1]{\small KTH Royal Institute of Technology, Stockholm, Sweden}
\affil[2]{\small Swedish NCSA, Swedish Armed Forces, Stockholm, Sweden}

\predate{\begin{center}}
\postdate{\end{center}\vspace{-3mm}}

\begin{document}
\maketitle

\begin{abstract}
  Ekerå and Håstad have introduced a variation of Shor's algorithm for the discrete logarithm problem~(DLP).
  Unlike Shor's original algorithm, Ekerå--Håstad's algorithm solves the short DLP in groups of unknown order.
  In this work, we prove a lower bound on the probability of Ekerå--Håstad's algorithm recovering the short logarithm~$d$ in a single run.
  By our bound, the success probability can easily be pushed as high as~$1 - 10^{-10}$ for any short~$d$.
  A key to achieving such a high success probability is to efficiently perform a limited search in the classical post-processing by leveraging meet-in-the-middle or random-walk techniques.
  These techniques may be generalized to speed up other related classical post-processing algorithms.
  Asymptotically, in the limit as the bit length~$m$ of~$d$ tends to infinity, the success probability tends to one if the limits on the search space are parameterized in~$m$.
  Our results are directly applicable to Diffie--Hellman in safe-prime groups with short exponents, and to RSA via a reduction from the RSA integer factoring problem~(IFP) to the short DLP.
\end{abstract}

\section{Introduction}
\label{sect:introduction}
Ekerå and Håstad have introduced a variation of Shor's algorithm for the discrete logarithm problem (DLP).
Unlike Shor's algorithm~\cite{shor94, shor97}, which computes general discrete logarithms in cyclic groups of known order, Ekerå--Håstad's algorithm~\cite{ekera-modifying, ekera-hastad, ekera-pp} computes short discrete logarithms in cyclic groups of unknown order.

Ekerå--Håstad's algorithm is cryptanalytically relevant in that it may be used to efficiently break finite-field Diffie--Hellman (FF-DH)~\cite{diffie-hellman} in safe-prime groups with short exponents~\cite{nistsp800-56a, rfc3526, rfc7919}.
It may furthermore be used to efficiently break the Rivest--Shamir--Adleman (RSA) cryptosystem~\cite{rsa}, via a reduction from the RSA integer factoring problem (IFP) to the short DLP in a group of unknown order.
For further details, see~\cite[Sect.~4]{ekera-hastad}, \cite[App.~A.2]{ekera-pp} and~\cite[Sect.~5.7.3]{ekera-phd-thesis}.

In their joint paper~\cite{ekera-hastad}, Ekerå and Håstad prove\footnote{In~\cite{ekera-hastad}, fix~$s = 1$:
The probability of observing a good pair~$(j, k)$ is at least~$1/8$ by~\cite[Lems.~2--3]{ekera-hastad}.
With probability at least~$3/4$, the lattice~$L$ has no very short vector by~\cite[Lem.~3]{ekera-hastad}, in which case we may enumerate vectors in~$L$ to efficiently find~$d$ given~$(j, k)$, see~\cite[Sect.~3.9]{ekera-hastad}.
The success probability in a single run is hence at least $3/32 = 9.375\%$.}~\cite[Lems.~1--3]{ekera-hastad} that the probability of their algorithm successfully recovering the logarithm~$d$ in a single run is at least $3/32 = 9.375\%$.
Ekerå and Håstad furthermore describe how tradeoffs may be made, between the number of runs that are required, and the amount of work that needs to be performed quantumly in each run.
These ideas parallel those of Seifert~\cite{seifert} for order finding.
When making tradeoffs in Ekerå--Håstad's algorithm with tradeoff factor~$s \ge 1$, at most $m(1+2/s)$ group operations need to be evaluated quantumly in each run, for~$m$ the bit length of~$d$.\footnote{By group operation, we mean an operation of the form $\ket{c, u} \rightarrow \ket{c, u \cdot v^c}$ in this context, for $u, v$ elements of a group (that is written multiplicatively) and~$c$ a control qubit.
The number of group operations that actually need to be evaluated quantumly may be reduced below what is stated here by means of optimizations such as windowing~\cite{gidney19, vmi05, van-meter-phd-thesis} (see also~\cite[Sect.~5.3.6.3]{ekera-phd-thesis}).}
After~$8s$ runs, Ekerå and Håstad~\cite{ekera-hastad} show that the probability of recovering~$d$ is at least $1 - 1/2^{s+1}$, if all subsets of~$s$ outputs from the $8s$~outputs are independently post-processed.\footnote{If at least~$s$ of the~$8s$ outputs are good pairs, which happens with probability at least~$1/2$.}

Ekerå~\cite{ekera-pp} later analyzed the probability distribution induced by the quantum algorithm in greater detail.
These insights allowed Ekerå to develop an improved classical post-processing algorithm in~\cite{ekera-pp}, which eliminates the requirement in~\cite{ekera-hastad} to perform $8s$~runs and to independently post-process all subsets of $s$~runs from the resulting set of $8s$~runs.
Instead, $\nruns \ge s$ runs may be performed and jointly post-processed classically.
Furthermore, Ekerå used the aforementioned insights to develop a simulator for the quantum algorithm.
The simulator allows the probability distribution induced by the quantum algorithm to be sampled when~$d$ is known.
In turn, this allows the number of runs~$\nruns$ required to recover~$d$ in the classical post-processing to be estimated by means of simulations, as a function of~$s$, $d$ and a lower bound on the success probability.

In particular, Ekerå shows in~\cite{ekera-pp} by means of simulations that a single run of the quantum algorithm is sufficient to recover~$d$ with success probability at least~$99\%$ when not making tradeoffs (i.e.~when $s \approx 1$) and performing at most~$3m$ group operations quantumly in the run, for~$m$ the bit length of~$d$.

\subsection{Our contributions}
\label{sect:contributions}
In this work, we improve on the state of the art by replacing the simulations-based part of the analysis in~\cite{ekera-pp} with a formal analysis that yields strict bounds.
This when not making tradeoffs and solving in a single run, in analogy with our formal analysis in~\cite{ekera-success} of the success probability of Shor's order-finding algorithm.

More specifically, we prove a lower bound on the probability of Ekerå--Håstad's algorithm recovering the short discrete logarithm~$d$ in a single run, and an associated upper bound on the complexity of the classical post-processing.

By our bounds, the success probability can easily be pushed as high as~$1 - 10^{-10}$ for any short~$d$.
This when performing at most~$3m$ group operations quantumly in the run, as in~\cite{ekera-pp}, for~$m$ the bit length of~$d$, and when requiring the classical post-processing to be feasible to perform in practice on an ordinary computer.
Further\-more, the number of group operations that need to be performed quantumly can be reduced below what is possible with the simulations-based analysis and post-processing in~\cite{ekera-pp} without compromising the practicality of the post-processing.\footnote{Note that ${3m = m(1 + 2/s)}$ when ${s = 1}$.
In practice, as explained in Sect.~\ref{sect:quantum-algorithm}, it suffices to perform ${m + 2 \ell}$ group operations when solving in a single run where ${\ell = m - \Delta}$ for small ${\Delta \in [0, m) \cap \mathbb Z}$.
Selecting ${\Delta = 0}$ then corresponds to ${s = 1}$, whereas selecting ${\Delta > 0}$ corresponds to~$s$ being slightly larger than one.}

A key to achieving these results is to efficiently perform a limited search in the classical post-processing by leveraging meet-in-the-middle or random-walk techniques.
These techniques may be generalized to speed up other related classical post-processing algorithms that perform limited searches, such as those in~\cite{ekera-pp, ekera-general, ekera-success, ekera-revisiting}, both when making and not making tradeoffs.
For further details, see App.~\ref{app:algorithms-generalizations}.

Asymptotically, in the limit as the bit length~$m$ of~$d$ tends to infinity, the success probability tends to one if the limits on the search space are parameterized in~$m$ so that the complexity of the post-processing grows at most polynomially in~$m$.

Our results are directly applicable to Diffie--Hellman in safe-prime groups with short exponents, and to RSA via a reduction from the RSA IFP to the short DLP.
For RSA, our bounds for the short DLP allow us to obtain better overall estimates of the success probability when using the aforementioned reduction.

\subsection{Overview of the introduction}
In the remainder of this introduction, we formally introduce the short DLP in Sect.~\ref{sect:short-dlp}, and then recall the quantum algorithm and the classical post-processing algorithm from~\cite{ekera-hastad, ekera-pp} in Sects.~\ref{sect:quantum-algorithm} and~\ref{sect:classical-post-processing}, respectively, whilst introducing notation.
We then proceed in Sect.~\ref{sect:overview} to give an overview of the remainder of this paper.

\subsection{The short discrete logarithm problem}
\label{sect:short-dlp}
In the short DLP, we are given a generator~$g$ of a cyclic group~$\langle g \rangle$ of order~$r$, where we assume in what follows that~$r$ is unknown, and $x = g^d$ for $d \lll r$ the logarithm, and are asked to compute~$d$.
Throughout this paper, we write~$\langle g \rangle$ multiplicatively.

\subsection{The quantum algorithm}
\label{sect:quantum-algorithm}
Let $m \in \mathbb Z$ be an upper bound\footnote{It suffices to use an upper bound on the bit length of~$d$ if the exact length is unknown.} on the bit length of the short discrete logarithm~$d$ so that~$d < 2^m$, and let $\ell = m - \Delta$ for small $\Delta \in [0, m) \cap \mathbb Z$.
The quantum algorithm in~\cite{ekera-hastad, ekera-pp} then induces the state
\begin{align}
  \label{eq:superposition}
  \frac{1}{2^{m+2\ell}}
  \sum_{a, \, j \, = \, 0}^{2^{m + \ell}-1}
  \sum_{b, \, k \, = \, 0}^{2^{\ell}-1}
  \exp\left(
    \frac{2 \pi \imag}{2^{m + \ell}} (aj + 2^m b k)
  \right)
  \ket{j, k, g^{a - bd}}
\end{align}
by using standard techniques, see Sect.~\ref{sect:quantum-algorithm-implementation} for further details.
When observed, the state~\refeq{superposition} yields a pair~$(j, k)$ and a group element~$g^e$ with probability
\begin{align}
  \label{eq:probability-a-b}
  \frac{1}{2^{2(m+2\ell)}}
  \left|\,
    \sum_{(a, b)}
    \exp\left(
      \frac{2 \pi \imag}{2^{m + \ell}} (aj + 2^m b k)
    \right)
  \,\right|^2
\end{align}
where the sum in~\refeq{probability-a-b} runs over all~$(a, b)$ such that ${a \in [0, 2^{m+\ell}) \cap \mathbb Z}$, ${b \in [0, 2^{\ell}) \cap \mathbb Z}$ and $e \equiv a - bd \:\: (\text{mod } r)$.
In what follows, as in~\cite{ekera-hastad, ekera-pp}, suppose that~$d$ is short in the sense that $r \ge 2^{m+\ell} + (2^\ell - 1) d$ so that $e = a - bd$.
Furthermore, as in~\cite{ekera-hastad, ekera-pp, ekera-general}, let
\begin{align}
  \alpha_d = \alpha(j, k) &= \{ dj + 2^m k \}_{2^{m+\ell}},
  &
  \theta_d = \theta(\alpha_d) &= \frac{2 \pi \alpha_d}{2^{m + \ell}},
\end{align}
be the argument and angle, respectively, yielded by the pair $(j, k)$, where~$\{ u \}_n$ denotes~$u$ reduced modulo~$n$ constrained to~$[-n/2, n/2)$.

Then, as shown in~\cite[Sect.~3]{ekera-pp}, by summing~\refeq{probability-a-b} over all~$e$, we have that the probability of observing a pair~$(j, k)$ yielding a given angle~$\theta_d$ is
\begin{align}
  P(\theta_d)
  =
  \frac{1}{2^{2(m+2\ell)}}
  \sum_{e \, = \, -(2^\ell - 1)d}^{2^{m+\ell}-1} \,
  \underbrace{\left|\,
    \sum_{b \, = \, 0}^{\#b(e) - 1}
    \e^{\imag \theta_d b}
  \,\right|^2}_{= \, \zeta(\theta_d, \#b(e))}, \label{eq:P}
\end{align}
where~$\#b(e)$ is the length of the contiguous range of values of~$b \in [0, 2^\ell) \cap \mathbb Z$ such that there exists~$a \in [0, 2^{m+\ell}) \cap \mathbb Z$ such that $e = a - bd$.

The classical post-processing recovers~$d$ from the pair~$(j, k)$ (see Sect.~\ref{sect:classical-post-processing}) so in practice the group element $g^e$ need not be observed; it may simply be discarded.

\subsubsection{Implementing the quantum algorithm}
\label{sect:quantum-algorithm-implementation}
As explained in~\cite[Sect.~3.3]{ekera-hastad} and in~\cite{shor94, shor97}, the state~\refeq{superposition} may be induced using standard techniques, by for instance first inducing uniform superpositions over ${a \in [0, 2^{m+\ell}) \cap \mathbb Z}$ and ${b \in [0, 2^{\ell}) \cap \mathbb Z}$, respectively, in the first two control registers,\footnote{This may be accomplished by independently initializing each qubit in the register to~$\ket{0}$ and applying a Hadamard~($\mathrm{H}$) gate to the qubit, leaving it in the state $\mathrm{H} \, \ket{0} = (\ket{0} + \ket{1}) / \sqrt{2}$.} and by then computing $g^a x^{-b} = g^{a - bd}$ to the third work register, yielding the state
\begin{align}
  \label{eq:superposition-intermediary}
  \frac{1}{\sqrt{2^{m+2\ell}}}
  \sum_{a \, = \, 0}^{2^{m + \ell}-1}
  \sum_{b \, = \, 0}^{2^{\ell}-1}
  \ket{a, b, g^{a - bd}}.
\end{align}

By applying quantum Fourier transforms~(QFTs) of size $2^{m+\ell}$ and $2^{\ell}$, respectively, in place to the first two control registers of~\refeq{superposition-intermediary}, the state~\refeq{superposition} is then obtained, allowing the pair $(j, k)$ to be observed by measuring the control registers.
In practice, the two exponentiations dominate the cost of inducing the state.

A quantum circuit that performs the above procedure is drawn in Fig.~\ref{fig:basic-circuit} in App.~\ref{app:figures}.
As illustrated in said figure, the exponentiations are performed by first pre-computing powers of two of~$g$ and~$x^{-1}$, respectively, classically, and then composing these powers quantumly conditioned on the control registers, by using that
\begin{align*}
  a = \sum_{i \, = \, 0}^{m+\ell-1} 2^i a_i,
  \quad
  b = \sum_{i \, = \, 0}^{\ell-1} 2^i b_i,
  \quad \Rightarrow \quad
  g^a = \prod_{i \, = \, 0}^{m+\ell-1} g^{2^i a_i},
  \quad
  x^{-b} = \prod_{i \, = \, 0}^{\ell-1} x^{-2^i b_i},
\end{align*}
where $a_i, \, b_i \in \{ 0, 1 \}$ are in quantum superpositions, and~$g^{2^i}$ and~$x^{-2^i}$ are classical constants.
To perform the compositions reversibly, powers of two of~$g^{-1}$ and~$x$ must typically also be pre-computed classically so as to enable uncomputation.
For this implementation approach to be efficient, it must hence be efficient not only to compose group elements quantumly, but also to invert group elements classically.

The short DLP is cryptographically relevant primarily for ${\langle g \rangle \subseteq \mathbb Z_{\Nmodulus}^*}$, for~$\Nmodulus$ a prime or composite.
In such groups, inverses may be computed efficiently via the extended Euclidean algorithm even if the order~$r$ of~$g$ is unknown.

\paragraph{Notes on optimizations}
The circuit in Fig.~\ref{fig:basic-circuit} may be optimized in various ways:
For instance, the QFT and the measurements that are performed with respect to the first control register may be moved left so that they are performed directly after the computation of~$g^a$ to the work register (as the first control register is left idle in between the aforementioned steps).
Analogously, the initialization of the second control register to a uniform superposition may be moved right so that it is performed directly before the computation of~$x^{-b}$ to the work register, see Fig.~\ref{fig:j-then-k-circuit} in App.~\ref{app:figures} for the resulting circuit.
As may be seen in said figure, this effectively implies that~$j$ is first computed, and that~$k$ is then computed given~$j$.
The space required to implement the two control registers is furthermore reduced, from ${m + 2 \ell}$~qubits to ${m + \ell}$~qubits, which is advantageous.

In practice, the above space-saving optimization may be taken further:
The state~\refeq{superposition} may be induced and the two control registers measured by leveraging the semi-classical QFT~\cite{griffiths-niu} with control qubit recycling~\cite{zalka, parker-plenio, mosca-ekert}.
In such an optimized circuit, the initialization of the uniform superpositions, the exponentiations, and the computation of the QFTs, are interleaved.
A single control qubit is repeatedly initialized to a uniform superposition $\mathrm{H} \, \ket{0}$, used to condition a composition with a classically pre-computed constant, and then transformed and measured, after which it is recycled in the next iteration.
This effectively implies that a single qubit suffices to implement the two control registers, and that~$j$ is first computed bit-by-bit after which~$k$ is computed bit-by-bit given~$j$.
See~\cite[Fig.~A.8 on p.~153]{ekera-phd-thesis} for a step-by-step visualization of how the operations in the circuit are interleaved.

Optimizations such as the semi-classical QFT with control qubit recycling are beyond the scope of this paper, but we mention them above in passing to highlight that it is standard practice to first compute~$j$, and to then compute~$k$ given~$j$.
We use this fact in our analysis, see the proof of Lem.~\ref{lemma:bound-tau-good-pair} in Sect.~\ref{sect:bound-tau-good-pair} below.

\paragraph{Other space-saving optimizations}
On the topic of space-saving optimizations, it should also be noted that Chevignard, Fouque and Schrottenloher~\cite{cfs25} have recently proposed an alternative implementation technique that leverages a residue number system and ideas from May and Schlieper~\cite{ms22} regarding compression robustness to compress the work register at the expense of not being able to recycle the control qubits.
When viewed through the prism of this implementation technique, Ekerå--Håstad's variation of Shor's algorithm achieves a space saving since the reduction it achieves in the number of group operations that need to be evaluated quantumly implies a corresponding reduction in the control register lengths, and hence in the number of control qubits that must be kept around when not being able to recycle them.

Increasing~$\Delta$ hence yields not only a reduction in the number of group operations that need to be evaluated quantumly, but also a space saving, when using this implementation technique, since there are $m + 2\ell = 3m - 2\Delta$ control qubits.
This implementation technique is at its most powerful when making tradeoffs, however, since the overall space usage can then be pushed down to $m + o(m)$ qubits at the expense of making $O(\log m)$ runs.

\subsection{The classical post-processing algorithm}
\label{sect:classical-post-processing}
As in~\cite{ekera-hastad, ekera-pp}, we use lattice-based techniques to classically recover~$d$ from~$(j, k)$, with a minor tweak to balance the lattice.
To describe the post-processing, it is convenient to introduce the below definition of a $\tau$-good pair~$(j, k)$:

\begin{definition}
  \label{def:tau-good-pair}
  For $\tau \in [0, \ell] \cap \mathbb Z$, a pair~$(j, k)$ is $\tau$-good if $|\, \{dj + 2^m k\}_{2^{m+\ell}} \,| \le 2^{m+\tau}$.
\end{definition}

It is furthermore convenient to introduce the lattice~$\mathcal L^\tau(j)$:

\begin{definition}
  Let~$\mathcal L^\tau(j)$ be the lattice generated by $(j, 2^\tau)$ and $(2^{m+\ell}, 0)$.
\end{definition}

If~$(j, k)$ is $\tau$-good, it follows that the known vector $\vec v = (\{-2^m k\}_{2^{m+\ell}}, 0) \in \mathbb Z^2$ is close to the unknown vector $\vec u = (dj + 2^{m+\ell} z, 2^\tau d) \in \mathcal L^\tau(j)$ for some $z \in \mathbb Z$.
More specifically, since $|\, \{dj + 2^m k\}_{2^{m+\ell}} \,| \le 2^{m+\tau}$ and $d < 2^m$, it holds that
\begin{align*}
  |\, \vec u - \vec v \,|
  &=
  \sqrt{( dj + 2^{m+\ell} z - \{-2^m k\}_{2^{m+\ell}} )^2 + (2^\tau d)^2} \\
  &=
  \sqrt{\{dj + 2^m k\}_{2^{m+\ell}}^2 + (2^\tau d)^2}
  <
  2^{m+\tau} \sqrt{2}.
\end{align*}

If~$(j, k)$ is $\tau$-good for small~$\tau$, then --- as explained in~\cite{ekera-hastad, ekera-pp} and Sect.~\ref{sect:analysis} --- the above implies that the unknown vector~$\vec u$ that yields~$d$ can be efficiently recovered by enumerating all vectors in~$\mathcal L^\tau(j)$ within a ball of radius ${2^{m+\tau} \sqrt{2}}$ centered on~$\vec v$, provided that~$\mathcal L^\tau(j)$ does not have an exceptionally short shortest non-zero vector.

Note that compared to the post-processing in~\cite{ekera-pp}, which works in~$\mathcal L^0(j)$, we balance the lattice by instead working in~$\mathcal L^\tau(j)$.
Furthermore, and more importantly, we leverage meet-in-the-middle or random-walk techniques to efficiently perform the enumeration, and we give a formal worst-case analysis that allows the enumeration complexity to be upper bounded, and the success probability to be lower bounded, as explained in Sect.~\ref{sect:contributions}.

\subsection{Notation}
\label{sect:notation}
In what follows, we let $\ceil{u}$, $\floor{u}$ and $\round{u}$ denote $u \in \mathbb R$ rounded up, down and to the closest integer, respectively.
Ties are broken by requiring that $\round{u} = u - \{ u \}_1$.

\subsection{Overview of the remainder of the paper}
\label{sect:overview}
In what follows in Sect.~\ref{sect:analysis} below, we lower bound the probability of the quantum algorithm yielding a $\tau$-good pair~$(j, k)$ in Lem.~\ref{lemma:bound-tau-good-pair}.
Furthermore, we lower bound the probability of the lattice~$\mathcal L^\tau(j)$ being $t$-balanced --- in the sense of it having a shortest non-zero vector of norm~$\lambda_1 \ge 2^{m-t}$ --- in Lem.~\ref{lemma:bound-t-balanced-Lj}.
In Sect.~\ref{sect:bound-enumeration-complexity}, we then proceed to upper bound the enumeration complexity of finding~$\vec u \in \mathcal L^\tau(j)$ and hence~$d$ given~$\vec v \in \mathbb Z^2$ when~$(j, k)$ is a $\tau$-good pair and~$\mathcal L^\tau(j)$ is $t$-balanced, in Lems.~\ref{lemma:bound-enumeration}--\ref{lemma:bound-enumeration-dlp}.
In Alg.~\ref{alg:recover-d} in App.~\ref{app:algorithms}, we give pseudocode for the enumeration algorithm analyzed in Lem.~\ref{lemma:bound-enumeration}, which uses meet-in-the-middle techniques.

In Sect.~\ref{sect:main-theorem}, we combine Lems.~\ref{lemma:bound-tau-good-pair}--\ref{lemma:bound-t-balanced-Lj} and Lems.~\ref{lemma:bound-enumeration} and~\ref{lemma:bound-enumeration-dlp} in our main theorems Thms.~\ref{thm:main}--\ref{thm:main-dlp}, so as to lower bound the probability of recovering~$d$ from~$(j, k)$ whilst upper bounding the enumeration complexity.
In Tabs.~\ref{tab:bounds}--\ref{tab:bounds-rsa} in App.~\ref{app:tables}, we tabulate the bounds in Thms.~\ref{thm:main}--\ref{thm:main-dlp}.
In App.~\ref{app:figures} we provide figures intended to facilitate reader comprehension.

\section{Bounding the success probability}
\label{sect:analysis}
Let us now proceed to bound the success probability as outlined above:

\subsection{Bounding the probability of observing a $\tau$-good pair}
\label{sect:bound-tau-good-pair}
\begin{lemma}
  \label{lemma:bound-tau-good-pair}
  For any given~$j$, the probability of observing~$k$ such that~$(j, k)$ is $\tau$-good is at least
  \begin{align*}
    1 - \psi'(2^\tau)
    >
    1 - \frac{1}{2^\tau} - \frac{1}{2 \cdot 2^{2\tau}} - \frac{1}{6 \cdot 2^{3\tau}},
  \end{align*}
  for $\tau \in [0, \ell] \cap \mathbb Z$, and for~$\psi'$ the trigamma function.
\end{lemma}
\begin{proof}
  As explained in Sect.~\ref{sect:quantum-algorithm-implementation} with reference to Figs.~\ref{fig:basic-circuit}--\ref{fig:j-then-k-circuit} in App.~\ref{app:figures}, the quantum algorithm that induces the state~\refeq{superposition} may be implemented in such a manner that~$j$ is first computed, and~$k$ is then computed given~$j$.

  By Cl.~\ref{claim:uniform-j} (see Sect.~\ref{sect:bound-tau-good-pair-supporting-claims} below), $j$ is then first selected uniformly at random from $[0, 2^{m+\ell}) \cap \mathbb Z$, after which~$k$ is selected from $[0, 2^\ell) \cap \mathbb Z$ given~$j$.
  Specifically, for~$P$ as in~\refeq{P}, $k$ is selected given~$j$ according to the probability distribution
  \begin{align}
    2^{m+\ell} \cdot P(\theta(\alpha(j, k)))
    &=
    \frac{1}{2^{m+3\ell}}
    \sum_{e \, = \, -(2^\ell - 1) d}^{2^{m+\ell} - 1}
    \zeta(\theta(\alpha(j, k)), \#b(e)), \label{eq:P-for-k-given-j}
  \end{align}
  where we recall that~$\zeta$ is defined in~\refeq{P}, $\theta$ and~$\alpha$ in Sect.~\ref{sect:quantum-algorithm}, and~$\#b(e)$ in~\cite[Sect.~3]{ekera-pp}.

  For each~$j \in [0, 2^{m+\ell}) \cap \mathbb Z$, there is a value~$k_0(j)$ of~$k \in [0, 2^\ell) \cap \mathbb Z$ such that
  \begin{align*}
    \alpha_{d, 0}(j)
    =
    \alpha(j, k_0(j))
    =
    \{ dj + 2^m k_0(j) \}_{2^{m+\ell}}
    =
    dj \text{ mod } 2^m
    \in [0, 2^m) \cap \mathbb Z.
  \end{align*}
  Let $k = (k_0(j) + t) \text{ mod } 2^\ell$ for $t \in [-2^{\ell-1}, 2^{\ell-1}) \cap \mathbb Z$.
  Then
  \begin{align*}
    \alpha(j, k)
    &=
    \{ dj + 2^m k \}_{2^{m+\ell}}
    =
    \{ dj + 2^m ((k_0(j) + t) \text{ mod } 2^\ell) \}_{2^{m+\ell}} \\
    &=
    \{ dj + 2^m (k_0(j) + t) \}_{2^{m+\ell}}
    =
    \{  \alpha_{d, 0}(j) + 2^m t \}_{2^{m+\ell}} \\
    &=
    \alpha_{d, 0}(j) + 2^m t \in [-2^{m+\ell-1}, 2^{m+\ell-1}) \cap \mathbb Z.
  \end{align*}

  For each~$j$, the probability of observing~$k$ such that
  \begin{align*}
    |\, \alpha(j, k) \,|
    &=
    |\, \alpha_{d, 0}(j) + 2^m t \,| \le 2^{m+\tau},
  \end{align*}
  i.e.~such that~$(j, k)$ is $\tau$-good, is then lower bounded by $1 - T_{+} - T_{-}$, for~$T_{+}$ and~$T_{-}$ the probability captured by the positive and negative tails, i.e.~by the regions where $t \in [2^\tau, 2^{\ell-1}) \cap \mathbb Z$ and $t \in [-2^{\ell-1}, -2^\tau) \cap \mathbb Z$, respectively, and where we have used that the probability distribution~\refeq{P-for-k-given-j} sums to one over~$k$ for fixed~$j$.

  It follows that we may lower bound the probability of observing a $\tau$-good pair~$(j, k)$ by upper bounding~$T_{+}$ and~$T_{-}$.
  More specifically
  \begin{align}
      T_{+}
    &=
      \frac{1}{2^{m+3\ell}}
      \sum_{t \, = \, 2^\tau}^{2^{\ell-1} - 1}
      \sum_{e \, = \, -(2^\ell - 1) d}^{2^{m+\ell} - 1}
      \zeta(\theta(\alpha(j, k_0(j) + t)), \#b(e)) \notag \\
    &\le
      \frac{1}{2^{m+3\ell}}
      \sum_{t \, = \, 2^\tau}^{2^{\ell-1} - 1}
      \frac{2^{m+\ell+1} \pi^2}{(\theta(\alpha(j, k_0(j) + t)))^2} \label{eq:T-plus-bound-zeta} \\
    &=
      \frac{1}{2^{m+3\ell}}
      \sum_{t \, = \, 2^\tau}^{2^{\ell-1} - 1}
      \frac{2^{m+\ell+1} \pi^2}{(2 \pi (\alpha_{d, 0}(j) + 2^m t) / 2^{m+\ell})^2}
    =
      \frac{1}{2}
      \sum_{t \, = \, 2^\tau}^{2^{\ell-1} - 1}
      \frac{2^{2m}}{(\alpha_{d, 0}(j) + 2^m t)^2} \notag \\
    &\le
      \frac{1}{2}
      \sum_{t \, = \, 2^\tau}^{2^{\ell-1} - 1}
      \frac{1}{t^2}
    <
      \frac{1}{2}
      \sum_{t \, = \, 2^\tau}^{\infty}
      \frac{1}{t^2}
    =
      \frac{1}{2} \psi'(2^\tau), \label{eq:T-plus-bound}
  \end{align}
  where we have used Cl.~\ref{claim:bound-zeta}, see Sect.~\ref{sect:bound-tau-good-pair-supporting-claims}, to bound~$\zeta$ in~\refeq{T-plus-bound-zeta}, and where~$\psi'$ in~\refeq{T-plus-bound} is the trigamma function.
  In~\refeq{T-plus-bound}, we have furthermore used that $\alpha_{d, 0}(j) \in [0, 2^m) \cap \mathbb Z$, and that the expression is maximized when $\alpha_{d, 0}(j) = 0$.
  Analogously
  \begin{align}
      T_{-}
    &=
      \frac{1}{2^{m+3\ell}}
      \sum_{t \, = \, -2^{\ell-1}}^{-2^\tau-1}
      \sum_{e \, = \, -(2^\ell - 1) d}^{2^{m+\ell} - 1}
      \zeta(\theta(\alpha(j, k_0(j) + t)), \#b(e)) \notag \\
    &\le
      \frac{1}{2^{m+3\ell}}
      \sum_{t \, = \, -2^{\ell-1}}^{-2^\tau-1}
      \frac{2^{m+\ell+1} \pi^2}{(\theta(\alpha(j, k_0(j) + t)))^2}
    =
      \frac{1}{2^{m+3\ell}}
      \sum_{t \, = \, 2^\tau+1}^{2^{\ell-1}}
      \frac{2^{m+\ell+1} \pi^2}{(\theta(\alpha(j, k_0(j) - t)))^2} \notag \\
    &=
      \frac{1}{2^{m+3\ell}}
      \sum_{t \, = \, 2^\tau+1}^{2^{\ell-1}}
      \frac{2^{m+\ell+1} \pi^2}{(2 \pi (\alpha_{d, 0}(j) - 2^m t) / 2^{m+\ell})^2}
    =
      \frac{1}{2}
      \sum_{t \, = \, 2^\tau+1}^{2^{\ell-1}}
      \frac{2^{2m}}{(\alpha_{d, 0}(j) - 2^m t)^2} \notag \\
    &\le
      \frac{1}{2}
      \sum_{t \, = \, 2^\tau+1}^{2^{\ell-1}}
      \frac{2^{2m}}{((2^m - 1) - 2^m t)^2} \label{eq:T-minus-bound-maximize-alpha0} \\
    &<
      \frac{1}{2}
      \sum_{t \, = \, 2^\tau+1}^{2^{\ell-1}}
      \frac{1}{(t - 1)^2}
    =
      \frac{1}{2}
      \sum_{t \, = \, 2^\tau}^{2^{\ell-1} - 1}
      \frac{1}{t^2}
    <
      \frac{1}{2}
      \sum_{t \, = \, 2^\tau}^{\infty}
      \frac{1}{t^2}
    =
      \frac{1}{2} \psi'(2^\tau), \label{eq:T-minus-bound}
  \end{align}
  where we have used in~\refeq{T-minus-bound-maximize-alpha0} that the expression is maximized when $\alpha_{d, 0}(j) = 2^m - 1$.

  It follows from~\refeq{T-plus-bound} and~\refeq{T-minus-bound} that the probability is lower bounded by
  \begin{align*}
    1 - T_{+} - T_{-}
    &>
    1 - \psi'(2^\tau)
    >
    1 - \frac{1}{2^\tau} - \frac{1}{2 \cdot 2^{2\tau}} - \frac{1}{6 \cdot 2^{3\tau}},
  \end{align*}
  where we have used Cl.~\ref{claim:bound-trigamma}, see Sect.~\ref{sect:bound-tau-good-pair-supporting-claims}, and so the lemma follows.
\end{proof}

\subsubsection{Supporting claims}
\label{sect:bound-tau-good-pair-supporting-claims}
The below claims support the proof of Lem.~\ref{lemma:bound-tau-good-pair} above:

\begin{claim}[{from \cite[Cl.~2.4]{ekera-success}}]
  \label{claim:bound-one-minus-cos}
  For any $\phi \in [-\pi, \pi]$, it holds that
  \begin{align*}
    \frac{2 \phi^2}{\pi^2} \le 1 - \cos(\phi) \le \frac{\phi^2}{2}.
  \end{align*}
\end{claim}
\begin{proof}
  This is a standard claim.
  Please see~\cite[Cl.~2.4]{ekera-success} for the proof.
\end{proof}

\begin{claim}
  \label{claim:bound-zeta}
  For $\zeta(\theta_d, \#b(e))$ the inner sum in~\refeq{P}, it holds that
  \begin{align*}
    \zeta(\theta_d, \#b(e)) \le \frac{\pi^2}{\theta_d^2}.
  \end{align*}
\end{claim}
\begin{proof}
  The claim trivially holds for $\theta_d = 0$.
  For $\theta_d \neq 0$, it holds that
  \begin{align*}
    \zeta(\theta_d, \#b(e))
    &=
    \left|\,
    \sum_{b \, = \, 0}^{\#b(e) - 1}
    \e^{\imag \theta_d b}
    \,\right|^2
    =
    \left|\,
    \frac{1 - \e^{\imag \theta_d \, \#b(e)}}{1 - \e^{\imag \theta_d}}
    \,\right|^2
    =
    \frac{1 - \cos(\theta_d \, \#b(e))}{1 - \cos(\theta_d)} \\
    &\le
    \frac{2}{1 - \cos(\theta_d)}
    \le
    \frac{2}{2 \theta_d^2 / \pi^2}
    =
    \frac{\pi^2}{\theta_d^2},
  \end{align*}
  where we have used that $|\, \theta_d \,| \le \pi$, and Cl.~\ref{claim:bound-one-minus-cos}, and so the claim follows.
\end{proof}

\begin{claim}[{from \cite[Cl.~3.2]{ekera-success}} via Nemes~\cite{nemes}]
  \label{claim:bound-trigamma}
  For any real $x > 0$, it holds that
  \begin{align*}
    \psi'(x) < \frac{1}{x} + \frac{1}{2x^2} + \frac{1}{6x^3}
  \end{align*}
  for~$\psi'$ the trigamma function.
\end{claim}
\begin{proof}
  Please see~\cite[Cl.~3.2]{ekera-success} for the proof.
\end{proof}

\begin{claim}
  \label{claim:uniform-j}
  The integer~$j$ yielded by the quantum algorithm that induces the state~\refeq{superposition} is selected uniformly at random from $[0, 2^{m+\ell}) \cap \mathbb Z$.
\end{claim}
\begin{proof}
  As explained in Sect.~\ref{sect:quantum-algorithm-implementation} with reference to Figs.~\ref{fig:basic-circuit}--\ref{fig:j-then-k-circuit} in App.~\ref{app:figures}, the quantum algorithm that induces the state~\refeq{superposition} may be implemented in such a manner that~$j$ is first computed, and~$k$ is then computed given~$j$.

  In the first step in Fig.~\ref{fig:j-then-k-circuit} where~$j$ is computed, the algorithm induces the state
  \begin{align}
    \label{eq:superposition-step-j}
    \frac{1}{2^{m+\ell}}
    \sum_{a, \, j \, = \, 0}^{2^{m + \ell}-1}
    \exp\left(
      \frac{2 \pi \imag}{2^{m + \ell}} aj
    \right)
    \ket{j, g^a}.
  \end{align}

  Note that no interference has yet arisen after this first step.
  Observing the first register in~\refeq{superposition-step-j} therefore yields each~$j \in [0, 2^{m+\ell}) \cap \mathbb Z$ with probability
  \begin{align*}
    \frac{1}{2^{2(m+\ell)}}
    \sum_{a \, = \, 0}^{2^{m+\ell} - 1}
    \underbrace{\left|\,
      \exp\left(
        \frac{2 \pi \imag}{2^{m + \ell}} aj
      \right)
    \,\right|^2}_{= \, 1}
    =
    \frac{1}{2^{m+\ell}}
  \end{align*}
  since $r > 2^{m+\ell}$ for~$r$ the order of~$g$ (this follows from the supposition in Sect.~\ref{sect:quantum-algorithm} that~$d$ is short in the sense that $r \ge 2^{m+\ell} + (2^\ell - 1) d$), and so the claim follows.
\end{proof}

\subsection{Bounding the probability of~$\mathcal L^\tau(j)$ being $t$-balanced}
\label{sect:bound-shortest-vector-Lj}
As in~\cite{ekera-success}, let~$\vec s_1$ be a shortest non-zero vector of~$\mathcal L^\tau(j)$, and let~$\vec s_2$ be a shortest non-zero vector that is linearly independent to~$\vec s_1$, up to signs.
Then~$(\vec s_1, \vec s_2)$ forms a Lagrange-reduced basis for~$\mathcal L^\tau(j)$.
It may be found efficiently by Lagrange's algorithm~\cite{lagrange, nguyen}.\footnote{Note that in the two-dimensional case that we consider in this paper, Lagrange's algorithm is equivalent to the Lenstra--Lenstra--Lovász (LLL) algorithm~\cite{lll} (with parameter ${\delta = 1}$) in practice.}
Let~$\vec s_2^{\parallel} = \mu \cdot \vec s_1$ and~$\vec s_2^{\perp} = \vec s_2 - \vec s_2^{\parallel}$ be the components of~$\vec s_2$ that are parallel and orthogonal to~$\vec s_1$, respectively, where $\mu = \langle \vec s_1, \vec s_2 \rangle / |\, \vec s_1 \,|^2$.
Furthermore, let $\lambda_1 = |\, \vec s_1 \,|$, $\lambda_2 = |\, \vec s_2 \,|$, $\lambda_2^{\perp} = |\, \vec s_2^{\perp} \,|$ and $\lambda_2^{\parallel} = |\, \vec s_2^{\parallel} \,|$.

\begin{claim}[{from \cite[Cl.~C.1]{ekera-success}}]
  \label{claim:lattice-det}
  It holds that $\lambda_1 \lambda_2^{\perp} = 2^{m+\ell+\tau}$.
\end{claim}
\begin{proof}
  This is a standard claim.
  It follows from the fact that the area of the fundamental region in~$\mathcal L^\tau(j)$ is $\lambda_1 \lambda_2^{\perp} = \det \mathcal L^\tau(j) = 2^{m+\ell+\tau}$.
\end{proof}

\begin{claim}[{from \cite[Cl.~C.2]{ekera-success}}]
  \label{claim:lattice-ineq}
  It holds that $\lambda_2^{\parallel} = |\, \mu \,| \cdot \lambda_1 \le \lambda_1 / 2$ and $\lambda_2^{\perp} \ge \sqrt{3} \, \lambda_2 / 2$.
\end{claim}
\begin{proof}
  This is a standard claim.
  Please see~\cite[Cl.~C.2]{ekera-success} for the proof.
  Note that this claim holds for any two-dimensional lattice, not only for~$\mathcal L^\tau(j)$.
\end{proof}

We are now ready to introduce the notion of~$\mathcal L^\tau(j)$ being $t$-balanced, and to bound the probability of $\mathcal L^\tau(j)$ not being $t$-balanced:

\begin{definition}
  For $t \in [0, m) \cap \mathbb Z$ and $\tau \in [0, \ell] \cap \mathbb Z$, the lattice $\mathcal L^\tau(j)$ is $t$-balanced if $\mathcal L^\tau(j)$ has a shortest non-zero vector of norm $\lambda_1 \ge 2^{m-t}$.
\end{definition}

\begin{lemma}
  \label{lemma:bound-t-balanced-Lj}
  The probability that~$\mathcal L^\tau(j)$ is not $t$-balanced is at most $2^{\Delta - 2(t - 1) - \tau}$.
\end{lemma}
\begin{proof}
  All vectors in~$\mathcal L^\tau(j)$ are of the form $(\omega j + 2^{m+\ell} z, 2^\tau \omega)$ for $\omega, z \in \mathbb Z$.
  Selecting~$z$ to minimize the absolute value of the first component yields $(\{ \omega j \}_{2^{m+\ell}}, 2^\tau \omega)$.

  For each $\omega \in ((-2^{m-t-\tau}, 2^{m-t-\tau}) \, \backslash \, \{ 0 \}) \cap \mathbb Z$, there are at most $2 \cdot 2^{m-t} - 1$ values of~$j$ such that $|\{ \omega j \}_{2^{m+\ell}}| < 2^{m-t}$.
  To see this, first note that $\omega \neq 0$ since~$\vec s_1$ is a shortest \emph{non-zero} vector.
  Second, note that as~$j$ runs through all integers on $[0, 2^{m+\ell})$, the expression~$\{ \omega j \}_{2^{m+\ell}}$ assumes (in some order) the values~$2^\kappa u$ for $u \in [-2^{m+\ell-\kappa-1}, 2^{m+\ell-\kappa-1}) \cap \mathbb Z$ a total of~$2^{\kappa}$ times, for~$2^{\kappa}$ the greatest power of two that divides~$\omega$.
  The worst case occurs when $\kappa = 0$, in which case each of the $2 \cdot 2^{m-t} - 1$ values in the range $(-2^{m-t}, 2^{m-t}) \cap \mathbb Z$ are assumed a single time.

  The number of~$j$ for which~$\mathcal L^\tau(j)$ has a shortest non-zero vector $\vec s_1 = (s_{1,1}, s_{1,2})$ such that $|\, s_{1,1} \,| < 2^{m-t}$ and $|\, s_{1,2} \,| < 2^{m-t}$ is hence upper bounded by
  \begin{align*}
    \max(0, \, 2 \cdot (2^{m - t - \tau} - 1)) \cdot (2 \cdot 2^{m - t} - 1)
    &<
    2^{m - t - \tau + 1} \cdot 2^{m - t + 1}
    =
    2^{2(m - t + 1) - \tau}.
  \end{align*}

  Since~$j$ is uniformly distributed on $[0, 2^{m+\ell}) \cap \mathbb Z$ by Cl.~\ref{claim:uniform-j}, see Sect.~\ref{sect:bound-tau-good-pair-supporting-claims}, where we recall that $\ell = m - \Delta$, the probability of observing~$j$ that is such that ${\lambda_1 = |\, \vec s_1 \,| < 2^{m-t}}$ is hence at most
  \begin{align*}
    2^{2(m - t + 1) - \tau} / 2^{m + \ell}
    &=
    2^{2m - 2t + 2 - \tau - 2m + \Delta}
    =
    2^{\Delta - 2(t - 1) - \tau},
  \end{align*}
  and so the lemma follows.
\end{proof}

\section{Bounding the enumeration complexity}
\label{sect:bound-enumeration-complexity}
We are now ready to bound the enumeration complexity when~$j$ is such that~$\mathcal L^\tau(j)$ is $t$-balanced, and when~$k$ given~$j$ is such that~$(j, k)$ is a $\tau$-good pair.

\subsection{Solving via a generalization of Shanks' algorithm}
\label{sect:bound-enumeration-complexity-shanks}
To start off, we explain in this section how to deterministically perform the enumeration by essentially generalizing Shanks' baby-step giant-step algorithm~\cite{shanks}, which leverages meet-in-the-middle time-memory tradeoff techniques, to two dimensions.

\begin{lemma}
  \label{lemma:bound-enumeration}
  Suppose that~$j$ is such that~$\mathcal L^\tau(j)$ is $t$-balanced, and that~$k$ given~$j$ is such that~$(j, k)$ is a $\tau$-good pair.
  Let ${\Nspace = 2^{\Delta+\tau+1} + 2^{\tau+t+2} + 2}$, and let~$c$ be a positive integer constant.
  Then at most $2^3 c \sqrt{\Nspace}$ group operations in $\langle g \rangle$ have to be performed to recover~$d$ from~$(j, k)$ by enumerating vectors in~$\mathcal L^\tau(j)$.
  This holds assuming that a few group elements are first pre-computed, and that there is space to store at most~$2^3 \sqrt{\Nspace} / c + 3$ integers in a lookup table.
\end{lemma}
\begin{proof}
  Recall that since~$(j, k)$ is $\tau$-good, it holds that $|\, \vec u - \vec v \,| < 2^{m+\tau} \sqrt{2}$, where $\vec v = (\{-2^m k\}_{2^{m+\ell}}, 0) \in \mathbb Z^2$, and~$\vec u$ is an unknown vector that yields~$d$, see Sect.~\ref{sect:classical-post-processing}.

  Let~$\vec o$ be the vector in~$\mathcal L^\tau(j)$ that is yielded by Babai's nearest plane algorithm upon input of~$\vec v$ and $(\vec s_1, \vec s_2)$.
  Then $\vec o - \vec v = \delta_1 \vec s_1 + \delta_2 \vec s_2^\perp$ where $|\, \delta_1 \,|, \, |\, \delta_2 \,| \le \frac{1}{2}$.

  To find~$\vec u$, it hence suffices to enumerate all vectors $\vec u'(m_1, m_2)$ of the form
  \begin{align*}
    \vec u'(m_1, m_2) = \vec o + (m_1 - \round{m_2 \cdot \mu}) \, \vec s_1 + m_2 \vec s_2
  \end{align*}
  for $m_1 \in [-B_1, B_1] \cap \mathbb Z$ and $m_2 \in [-B_2, B_2] \cap \mathbb Z$, respectively, where
  \begin{align*}
    B_1 = \lfloor 2^{m+\tau} \sqrt{2} / \lambda_1 + 1 \rfloor
    \quad \text{ and } \quad
    B_2 = \lfloor 2^{m+\tau} \sqrt{2} / \lambda_2^{\perp} + 1/2 \rfloor.
  \end{align*}

  To see this, note that there are at most ${2 B_2 + 1}$ values of~$m_2$ to explore to find a point ${\vec o + m_2 \vec s_2}$ on the line parallel to~$\vec s_1$ on which the vector~${\vec u \in \mathcal L^\tau(j)}$ lies.
  There are at most ${2 B_1 + 1}$ vectors to explore on each of these lines to find~$\vec u$.
  Note that the ``off-drift''\label{page:off-drift} in the direction of~$\vec s_1$ when adding ${m_2 \vec s_2}$ to~$\vec o$ is compensated for by at the same time subtracting ${\round{m_2 \cdot \mu} \, \vec s_1}$.
  Furthermore, note that
  \begin{align*}
    |\, \vec u'(m_1, m_2) - \vec v \,|^2
    &=
    |\, \vec o + (m_1 - \round{m_2 \cdot \mu}) \, \vec s_1 + m_2 \vec s_2 - \vec v \,|^2 \\
    &=
    |\, \delta_1 \vec s_1 + \delta_2 \vec s_2^\perp + (m_1 - \round{m_2 \cdot \mu}) \, \vec s_1 + m_2 (\vec s_2^{\parallel} + \vec s_2^{\perp}) \,|^2 \\
    &=
    |\, (m_1 + \delta_1 + m_2 \cdot \mu - \round{m_2 \cdot \mu}) \, \vec s_1 \,|^2 + |\, (m_2 + \delta_2) \, \vec s_2^\perp \,|^2
  \end{align*}
  as $\vec s_2^{\parallel} = \mu \vec s_1$, where it suffices to let $B_2 = \floor{2^{m+\tau} \sqrt{2} / \lambda_2^{\perp} + 1/2}$ since
  \begin{align*}
    (|\, m_2 \,| - 1/2) \, \lambda_2^{\perp}
    &\le
    (|\, m_2 \,| - \underbrace{|\, \delta_2 \,|}_{\le \frac{1}{2}}) \, \lambda_2^{\perp}
    \le
    |\, (m_2 + \delta_2) \, \vec s_2^{\perp} \,|
    <
    2^{m+\tau} \sqrt{2}.
  \end{align*}

  Analogously, it suffices to let $B_1 = \floor{2^{m+\tau} \sqrt{2} / \lambda_1 + 1}$ since
  \begin{align*}
    (|\, m_1 \,| - 1) \, \lambda_1
    &\le
    (|\, m_1 \,| - \underbrace{|\, \delta_1 \,|}_{\le \frac{1}{2}} - \underbrace{|\, m_2 \cdot \mu - \round{m_2 \cdot \mu} |}_{\le \frac{1}{2}}) \, \lambda_1 \\
    &\le
    |\, (m_1 + \delta_1 + m_2 \cdot \mu - \round{m_2 \cdot \mu}) \, \vec s_1 \,|
    <
    2^{m+\tau} \sqrt{2}.
  \end{align*}

  By Cl.~\ref{claim:bounds-B1-B2} and Lem.~\ref{lemma:bound-enumeration-generic} below, at most $2^3 c \sqrt{B_1 (B_2 + 1)}$ group operations in~$\langle g \rangle$ have to be performed to enumerate the aforementioned vectors in~$\mathcal L^\tau(j)$, and to test if the last component yields~$d$.
  This holds assuming that a few group elements are first pre-computed, and that there is space to store at most $2^3 \sqrt{B_1 (B_2 + 1)} \big/ c + 3$ integers in a lookup table.

  It furthermore holds that
  \begin{align*}
      B_1 (B_2 + 1)
      &=
      \lfloor 2^{m+\tau} \sqrt{2} / \lambda_1 + 1 \rfloor \cdot ( \lfloor 2^{m+\tau} \sqrt{2} / \lambda_2^{\perp} + 1/2 \rfloor + 1 ) \\
      &\le
      (2^{m+\tau} \sqrt{2} / \lambda_1 + 1) \cdot ( 2^{m+\tau} \sqrt{2} / \lambda_2^{\perp} + 3/2 ) \\
      &=
      2^{2(m+\tau)+1} / (\lambda_1 \lambda_2^{\perp}) + 2^{m+\tau} \sqrt{2} \, (3 / (2 \lambda_1) + 1 / \lambda_2^{\perp}) + 3/2 \\
      &\le
      2^{2(m+\tau)+1} / 2^{m+\ell+\tau} + 2^{\tau+t} \underbrace{\sqrt{2} \, (3 / 2 + 2 / \sqrt{3})}_{< \, 2^2} + \underbrace{3/2}_{< \, 2} \\
      &<
      2^{\Delta+\tau+1} + 2^{\tau+t+2} + 2 = \Nspace
  \end{align*}
  where we have used that ${\lambda_1 \ge 2^{m-t}}$, that $\lambda_2^{\perp} \ge \sqrt{3} \, \lambda_2 / 2 \ge \sqrt{3} \, \lambda_1 / 2$ by Cl.~\ref{claim:lattice-ineq}, and that ${\lambda_1 \lambda_2^{\perp} = 2^{m+\ell+\tau}}$ by Cl.~\ref{claim:lattice-det}, and so the lemma follows.
\end{proof}

\begin{claim}
  \label{claim:bounds-B1-B2}
  It holds that $B_1 \ge 1$ and $2 B_1 > B_2 \ge 0$ when
  \begin{align*}
    B_1 &= \lfloor 2^{m+\tau} \sqrt{2} / \lambda_1 + 1 \rfloor,
    &
    B_2 &= \lfloor 2^{m+\tau} \sqrt{2} / \lambda_2^{\perp} + 1/2 \rfloor.
  \end{align*}
\end{claim}
\begin{proof}
  Trivially $B_2 \ge 0$ and $B_1 = \lfloor 2^{m+\tau} \sqrt{2} / \lambda_1 \rfloor + 1 \ge 1$.
  Furthermore,
  \begin{align*}
    B_1
    &=
    \lfloor 2^{m+\tau} \sqrt{2} / \lambda_1 + 1 \rfloor
    >
    2^{m+\tau} \sqrt{2} / \lambda_1, \\
    B_2
    &=
    \lfloor 2^{m+\tau} \sqrt{2} / \lambda_2^{\perp} + 1/2 \rfloor
    \le
    2^{m+\tau} \sqrt{2} / \lambda_2^{\perp} + 1/2,
  \end{align*}
  where we have used that $f \ge \floor{f} > f - 1$ for $f \in \mathbb R$.
  Hence, it holds that
  \begin{align*}
    \frac{B_2}{B_1}
    &\le
    \frac{2^{m+\tau} \sqrt{2} / \lambda_2^{\perp}}{B_1}
    +
    \frac{1}{2 B_1}
    <
    \frac{2^{m+\tau} \sqrt{2} / \lambda_2^{\perp}}{2^{m+\tau} \sqrt{2} / \lambda_1}
    +
    \frac{1}{2}
    =
    \frac{\lambda_1}{\lambda_2^{\perp}} + \frac{1}{2}
    \le
    \frac{2}{\sqrt{3}} + \frac{1}{2} < 2
  \end{align*}
  since $\lambda_2^{\perp} \ge \sqrt{3} \, \lambda_2 / 2 \ge \sqrt{3} \, \lambda_1 / 2$, see Cl.~\ref{claim:lattice-ineq}, and so the claim follows.
\end{proof}

\begin{lemma}
  \label{lemma:bound-enumeration-generic}
  Let $\vec o \in \mathcal L^\tau(j)$, let $B_1 \ge 1$ and~$B_2 \ge 0$ be integers such that $2 B_1 > B_2$, and let~$c$ be a positive integer constant.
  Then, to enumerate the ${(2 B_1 + 1) (2 B_2 + 1)}$ vectors given by
  \begin{align*}
    {\vec o + (m_1 - \round{m_2 \cdot \mu}) \, \vec s_1 + m_2 \vec s_2}
  \end{align*}
  where
  ${m_1 \in [-B_1, B_1] \cap \mathbb Z}$
  and
  ${m_2 \in [-B_2, B_2] \cap \mathbb Z}$,
  and to test if $x = g^d$ for~$2^\tau d$ the last component of the vector, at most ${2^3 c \sqrt{B_1 (B_2 + 1)}}$ group operations in~$\langle g \rangle$ have to be performed.
  This holds assuming that a few group elements are first pre-computed, and that there is space to store at most ${2^3 \sqrt{B_1 (B_2 + 1)} \big/ c + 3}$ integers in a lookup table.
\end{lemma}
\begin{proof}
  Let $\vec o = \nu_1 \vec s_1 + \nu_2 \vec s_2$ for $\nu_1, \nu_2 \in \mathbb Z$.
  Let $\vec s_1 = (s_{1,1}, s_{1,2})$, $\vec s_2 = (s_{2,1}, s_{2,2})$.
  Let $s_1 = s_{1,2} / 2^\tau$ and $s_2 = s_{2,2} / 2^\tau$.
  Note that~$s_{1,2}$ and~$s_{2,2}$ are both divisible by~$2^\tau$ by design, as a consequence of how the basis for $\mathcal L^\tau(j)$ is setup, so $s_1, s_2 \in \mathbb Z$.

  Let $n = c \, \lfloor \sqrt{B_1 / (B_2 + 1)} \rceil$ for reasons that will be further elaborated on below, and perform a meet-in-the-middle search in two stages as outlined below:

  First compute~$g^{n \cdot i \cdot s_{1}}$ as~$i$ runs all over $[-{\Nspace}_1, {\Nspace}_1] \cap \mathbb Z$ for ${\Nspace}_1 = \ceil{B_1 / n}$.
  Insert the resulting ${2 {\Nspace}_1 + 1} = {2 \ceil{B_1 / n} + 1}$ integers~$i$ into a lookup table~$T$ indexed by~$g^{n \cdot i \cdot s_{1}}$.
  Then, compute~$g^{(\nu_1 + i - \round{j \cdot \mu}) \cdot s_{1} + (\nu_2 + j) \cdot s_{2}} \cdot x^{-1}$ for all combinations of~$i$ and~$j$, as~$i$ runs over $[0, n) \cap \mathbb Z$ and~$j$ over $[-B_2, B_2] \cap \mathbb Z$.
  For each resulting element, check if it indexes an integer~$k$ in~$T$:
  If so, $d = (\nu_1 + i - \round{\, j \cdot \mu} - k \cdot n) s_{1} + (\nu_2 + j) s_{2}$.

  The above two-stage search may be implemented efficiently, so that only
  \begin{align*}
    2 (\ceil{B_1 / n} - 1)
    <
    2 ((B_1 / n + 1) - 1)
    =
    2 B_1 / n
  \end{align*}
  group operations have to be performed in the first stage, and
  \begin{align*}
    2 B_2 + 2 (B_2 + 1) (n - 1)
    &=
    2 ((B_2 + 1)n - 1)
    <
    2 (B_2 + 1) n
  \end{align*}
  in the second stage, provided that the elements ${g_1 = g^{s_{1}}}$, ${g_1^{-1}}$, ${s = g_1^n}$, ${s^{-1}}$, ${g_2 = g^{s_{2}}}$, ${g_2^{-1}}$ and ${w = g_1^{\nu_1} \cdot g_2^{\nu_2} \cdot x^{-1}}$, and the combinations ${g_2 \cdot g_1}$, ${g_2 \cdot g_1^{-1}}$, ${g_2^{-1} \cdot g_1}$ and ${g_2^{-1} \cdot g_1^{-1}}$, are pre-computed.
  For the full details, see Alg.~\ref{alg:recover-d} in App.~\ref{app:algorithms}.
  Above, we picked $n = c \, \lfloor \sqrt{B_1 / (B_2 + 1)} \rceil$ to have $B_1 / n \approx (B_2 + 1) n$ when $c = 1$.
  When $c > 1$, we store a factor~$\sim c$ fewer integers in~$T$, and perform a factor~$\sim c$ less work in the first stage, at the expense of performing a factor~$c$ more work in the second stage.

  \vspace{1ex}

  \emph{Case I:} Suppose that $B_1 \ge B_2$: Then $B_1 \ge B_2 \ge 0$ and furthermore $B_1 \ge 1$.
  The number of group operations performed in total is then at most
  \begin{align*}
    2 B_1 / n + 2 (B_2 + 1) n
    &=
    \frac{2 B_1}{c \, \lfloor \sqrt{B_1 / (B_2 + 1)} \rceil} + 2c (B_2 + 1) \lfloor \sqrt{B_1 / (B_2 + 1)} \rceil \\
    &=
    \frac{2 B_1}{c (\sqrt{B_1 / (B_2 + 1)} + \delta)} + 2c (B_2 + 1) (\sqrt{B_1 / (B_2 + 1)} + \delta) \\
    &=
    \frac{2 B_1}{c \sqrt{B_1 / (B_2 + 1)}(1 + \delta')} + 2c (B_2 + 1) \sqrt{B_1 / (B_2 + 1)}(1 + \delta') \\
    &=
    2 \sqrt{B_1 (B_2 + 1)} \underbrace{\left( \frac{1}{c(1 + \delta')} + c(1 + \delta') \right)}_{< 4c}
    <
    2^3 c \sqrt{B_1 (B_2 + 1)},
  \end{align*}
  for some $\delta \in (-1/2, 1/2]$ and $\delta' = \delta / \sqrt{B_1 / (B_2 + 1)}$.
  Note that since $B_1 \ge B_2$, we have that $\sqrt{B_1 / (B_2 + 1)} \ge 1/\sqrt{2}$, and hence that $\delta' \in (-1/\sqrt{2}, 1/\sqrt{2}]$.
  In the last step, we maximize the expression by letting $\delta' = -1/\sqrt{2}$.

  As for the space usage, the number of integers stored in~$T$ is
  \begin{align*}
    2 \ceil{B_1 / n} + 1
    &\le
    2 (B_1 / n + 1) + 1
    =
    2 B_1 / n + 3 \\
    &=
    \frac{2 B_1}{c \, \lfloor \sqrt{B_1 / (B_2 + 1)} \rceil} + 3
    =
    \frac{2 B_1}{c (\sqrt{B_1 / (B_2 + 1)} + \delta)} + 3 \\
    &=
    \frac{2 B_1}{c \sqrt{B_1 / (B_2 + 1)}(1 + \delta')} + 3
    =
    \frac{2 \sqrt{B_1 (B_2 + 1)}}{c(1 + \delta')} + 3 \\
    &<
    2^3 \sqrt{B_1 (B_2 + 1)} \big/ c + 3,
  \end{align*}
  where we again maximize the expression by letting $\delta' = -1/\sqrt{2}$ in the last step.

  \vspace{1ex}

  \emph{Case II:} Suppose instead that $B_1 < B_2$:
  Then $1 \le B_1 < B_2 < 2 B_1$, so
  \begin{align*}
    B_1 &= \sqrt{B_1^2} < \sqrt{B_1 B_2} < \sqrt{B_1 (B_2 + 1)}, \\
    B_2 + 1 &= \sqrt{(B_2 + 1)^2} \le \sqrt{2 B_1 (B_2 + 1)},
  \end{align*}
  and $n = c \, \lfloor \sqrt{B_1 / (B_2 + 1)} \rceil = c \ge 1$ since
  \begin{align*}
    1 / \sqrt{3} \le \sqrt{B_1 / (2 B_1 + 1)} < \sqrt{B_1 / (B_2 + 1)} < \sqrt{B_2 / (B_2 + 1)} < 1.
  \end{align*}

  The number of group operations that have to be performed is hence at most
  \begin{align*}
    2 B_1 / n + 2 (B_2 + 1) n
    &=
    2 B_1 / c + 2 (B_2 + 1) c
    \le
    2 c (B_1 + (B_2 + 1)) \\
    &<
    2 c (1 + \sqrt{2}) \sqrt{B_1 (B_2 + 1)}
    <
    2^3 c \sqrt{B_1 (B_2 + 1)},
  \end{align*}
  and the number of integers stored in~$T$ is then
  \begin{align*}
    2 \ceil{B_1 / n} + 1
    &=
    2 \ceil{B_1 / c} + 1
    \le
    2 (B_1 / c + 1) + 1
    =
    2 B_1 / c + 3 \\
    &<
    2 \sqrt{B_1 (B_2 + 1)} / c + 3
    <
    2^3 \sqrt{B_1 (B_2 + 1)} / c + 3.
  \end{align*}

  \vspace{1ex}

  The total number of group operations is hence at most ${2^3 c \sqrt{B_1 (B_2 + 1)}}$ and the number of integers stored in~$T$ is at most ${2^3 \sqrt{B_1 (B_2 + 1)} / c + 3}$ irrespective of whether ${B_1 \ge B_2}$ or ${B_1 < B_2}$, and so the lemma follows.
\end{proof}

\noindent The full enumeration algorithm is described in pseudocode in Alg.~\ref{alg:recover-d} in App.~\ref{app:algorithms-shanks}.

\subsection{Solving via Gaudry--Schost's algorithm}
\label{sect:bound-enumeration-complexity-gaudry-schost}
When solving $(j, k)$ for~$d$ by generalizing Shanks' algorithm~\cite{shanks} to two dimensions as in Lem.~\ref{lemma:bound-enumeration} in the previous section, the space usage is typically a limiting factor when attempting to select large $\Delta$ and/or large~$t$ and~$\tau$.

A good option for reducing the space usage from $O(\sqrt{\Nspace})$~lookup table entries (for $\Nspace$ as in Lem.~\ref{lemma:bound-enumeration}) down to $O(1)$~group elements is to instead solve~$(j, k)$ for~$d$ by rewriting the enumeration problem in~$\mathcal L^\tau(j)$ as a two-dimensional short DLP and solving it by generalizing Pollard's $\lambda$-algorithm~\cite{pollard-rho-lambda, oorschot-wiener} to two dimensions.
Note that this is in analogy with how Pollard reduced the space usage in Shanks' algorithm back in the 1970s, in the one-dimensional case, by substituting the deterministic meet-in-the-middle techniques that Shanks uses with probabilistic random-walk techniques.
The two-dimensional case is trickier, however, since cycles can then e.g.\ arise in the random walks.

In earlier works, Gaudry and Schost~\cite{gaudry-schost} have explored generalizations of Pollard's $\lambda$-algorithm to two dimensions in the context of framing other problems as two-dimensional short DLPs.
Galbraith and Ruprai~\cite{galbraith-ruprai} have in turn improved Gaudry--Schost's algorithm, and generalized it to higher dimensions than two.
In this section, we give a variation of Lem.~\ref{lemma:bound-enumeration} that uses Gaudry--Schost's algorithm with Galbraith--Ruprai's improvements to solve~$(j, k)$ for~$d$ by writing the problem as a short two-dimensional DLP.

\begin{lemma}[Variation of Lem.~\ref{lemma:bound-enumeration}]
  \label{lemma:bound-enumeration-dlp}
  Suppose that~$j$ is such that~$\mathcal L^\tau(j)$ is $t$-balanced, and that~$k$ given~$j$ is such that~$(j, k)$ is a $\tau$-good pair.
  Let ${\Nspace = 2^{\Delta+\tau+4} + 2^{\tau+t+5} + 5}$.
  Then, the expected number of group operations required to solve~$(j, k)$ for~$d$ by reducing the enumeration problem in~$\mathcal{L}^\tau(j)$ to a two-dimensional short DLP and solving it via Gaudry--Schost's algorithm~\cite{gaudry-schost}, as generalized and improved by Galbraith and Ruprai~\cite{galbraith-ruprai}, in the idealized model, in the best, average and worst cases, is at most $(4/3 + o(1)) \, \sqrt{\pi \Nspace}$.
  This holds assuming that a few group elements are first pre-computed.
\end{lemma}
\begin{proof}
  Recall that since~$(j, k)$ is $\tau$-good, it holds that $|\, \vec u - \vec v \,| < 2^{m+\tau} \sqrt{2}$, where $\vec v = (\{-2^m k\}_{2^{m+\ell}}, 0) \in \mathbb Z^2$, and~$\vec u$ is an unknown vector that yields~$d$, see Sect.~\ref{sect:classical-post-processing}.

  Let~$\vec o$ be the vector in~$\mathcal L^\tau(j)$ that is yielded by Babai's nearest plane algorithm upon input of~$\vec v$ and $(\vec s_1, \vec s_2)$.
  Then $\vec o - \vec v = \delta_1 \vec s_1 + \delta_2 \vec s_2^\perp$ where $|\, \delta_1 \,|, \, |\, \delta_2 \,| \le \frac{1}{2}$.

  To find~$\vec u$, we enumerate all vectors of the form $\vec o + m_1 \vec s_1 + m_2 \vec s_2$ such that
  \begin{align}
    |\, \vec o + m_1 \vec s_1 + m_2 \vec s_2 - \vec v \,| < 2^{m+\tau} \sqrt{2} \label{eq:vector-form}
  \end{align}
  where $m_1 \in [-B_1, B_1] \cap \mathbb Z$ and $m_2 \in [-B_2, B_2] \cap \mathbb Z$, respectively.
  To upper bound $B_1$ and $B_2$, we use that $\vec s_2 = \vec s_2^{\parallel} + \vec s_2^{\perp}$ to write~\refeq{vector-form} as
  \begin{align*}
    |\, \vec o + m_1 \vec s_1 + m_2 \vec s_2 - \vec v \,|
    &=
    |\, \delta_1 \vec s_1 + \delta_2 \vec s_2^{\perp} + m_1 \vec s_1 + m_2 \vec s_2 \,| \\
    &=
    |\, \delta_1 \vec s_1 + \delta_2 \vec s_2^{\perp} + m_1 \vec s_1 + m_2 (\vec s_2^{\parallel} + \vec s_2^{\perp}) \,| \\
    &=
    |\, (m_1 + \delta_1) \, \vec s_1 + m_2 \, \vec s_2^{\parallel} + (m_2 + \delta_2) \, \vec s_2^{\perp} \,| < 2^{m+\tau} \sqrt{2}
  \end{align*}
  which, since~$\vec s_1$ is parallel to~$\vec s_2^{\parallel} = \mu \vec s_1$, and orthogonal to~$\vec s_2^{\perp}$, in turn yields
  \begin{align*}
    |\, \vec o + m_1 \vec s_1 + m_2 \vec s_2 - \vec v \,|^2
    &=
    |\, (m_1 + \delta_1) \, \vec s_1 \, + m_2 \, \vec s_2^{\parallel} |^2 + |\, (m_2 + \delta_2) \, \vec s_2^{\perp} \,|^2 \\
    &=
    |\, (m_1 + \mu \cdot m_2 + \delta_1) \, \vec s_1 |^2 + |\, (m_2 + \delta_2) \, \vec s_2^{\perp} \,|^2 \\
    &<
    2^{2(m+\tau) + 1}
  \end{align*}
  which implies that we may upper bound~$B_2$ based on the second term above, and then upper bound~$B_1$ based on $B_2$ and the first term above, yielding
  \begin{align*}
    B_2 \le \left\lfloor \frac{2^{m+\tau}}{\lambda_2^{\perp}} \sqrt{2} + \frac{1}{2} \right\rfloor,
    \:\:
    B_1 \le \left\lfloor \frac{2^{m+\tau}}{\lambda_1} \sqrt{2} + |\, \mu \,| \cdot B_2 + \frac{1}{2} \right\rfloor \le \left\lfloor \frac{2^{m+\tau}}{\lambda_1} \sqrt{2} + \frac{B_2}{2} + \frac{1}{2} \right\rfloor,
  \end{align*}
  where we have used that $\lambda_2^{\parallel} = |\, \mu \,| \cdot \lambda_1$ where $|\, \mu \,| \le 1 / 2$, see Cl.~\ref{claim:lattice-ineq}.

  To write the above enumeration as a two-dimensional short DLP, first let ${\vec s_1 = (s_{1,1}, s_{1,2})}$ and ${\vec s_2 = (s_{2,1}, s_{2,2})}$, then let ${s_1 = s_{1,2} / 2^\tau}$ and ${s_2 = s_{2,2} / 2^\tau}$, and finally let ${g_1 = g^{s_1}}$ and ${g_2 = g^{s_2}}$.
  Furthermore, let ${\vec o = \nu_1 \vec s_1 + \nu_2 \vec s_2}$.
  Then, for some ${i_1, i_2}$ such that ${i_1 \in [-B_1, B_1] \cap \mathbb Z}$ and ${i_2 \in [-B_2, B_2] \cap \mathbb Z}$, respectively, it holds that
  \begin{align*}
    g^d
    =
    g^{(\nu_1 + i_1) s_1 + (\nu_2 + i_2) s_2}
    =
    g_1^{\nu_1 + i_1} \, g_2^{\nu_2 + i_2}
    =
    x
    \quad\Rightarrow\quad
    g_1^{i_1} \, g_2^{i_2}
    =
    x \, g_1^{-\nu_1} \, g_2^{-\nu_2}
    = x'
  \end{align*}
  so we may solve $g_1^{i_1} \, g_2^{i_2} = x'$ for $i_1, i_2$, and then compute $d = (\nu_1 + i_1) s_1 + (\nu_2 + i_2) s_2$.

  We may furthermore simplify the upper bound on $B_1$, by using that $\lambda_2^{\perp} \ge \sqrt{3} \, \lambda_2 / 2 \ge \sqrt{3} \, \lambda_1 / 2$, see again Cl.~\ref{claim:lattice-ineq}, yielding
  \begin{align*}
    B_1 &\le \left\lfloor \frac{2^{m+\tau}}{\lambda_1} \sqrt{2} + \frac{B_2}{2} + \frac{1}{2} \right\rfloor
         \le \left\lfloor \frac{2^{m+\tau}}{\lambda_1} \sqrt{2} + \frac{2^{m+\tau}}{\lambda_2^{\perp}} \frac{1}{\sqrt{2}} + \frac{3}{4} \right\rfloor \\
        &\le \left\lfloor \frac{2^{m+\tau}}{\lambda_1} \sqrt{2} + \frac{2^{m+\tau}}{\lambda_1} \sqrt{\frac{2}{3}} + \frac{3}{4} \right\rfloor
         =   \left\lfloor \frac{2^{m+\tau}}{\lambda_1} \sqrt{2} \left( 1 + \frac{1}{\sqrt{3}} \right) + \frac{3}{4} \right\rfloor
  \end{align*}
  which implies that
  \begin{align*}
    &\phantom{\le} \:\:\: (2 B_1 + 1) (2 B_2 + 1) \\
    &\le
    \left( 2 \left\lfloor \frac{2^{m+\tau}}{\lambda_1} \sqrt{2} \left( 1 + \frac{1}{\sqrt{3}} \right) + \frac{3}{4} \right\rfloor + 1 \right)
    \left( 2 \left\lfloor \frac{2^{m+\tau}}{\lambda_2^{\perp}} \sqrt{2} + \frac{1}{2} \right\rfloor + 1 \right) \\
    &\le
    \left( \frac{2^{m+\tau+1}}{\lambda_1} \sqrt{2} \left( 1 + \frac{1}{\sqrt{3}} \right) + \frac{5}{2} \right)
    \left( \frac{2^{m+\tau+1}}{\lambda_2^{\perp}} \sqrt{2} + 2 \right) \\
    &=
    \frac{2^{2(m+\tau)+3}}{\lambda_1 \lambda_2^{\perp}} \left( 1 + \frac{1}{\sqrt{3}} \right)
    +
    \frac{2^{m+\tau+2}}{\lambda_1} \sqrt{2} \left( 1 + \frac{1}{\sqrt{3}} \right)
    +
    \frac{2^{m+\tau+1}}{\lambda_2^{\perp}} \frac{5}{\sqrt{2}}
    +
    5 \\
    &\le
    \frac{2^{2(m+\tau)+3}}{\lambda_1 \lambda_2^{\perp}}
    \underbrace{\left(
      1
      +
      \frac{1}{\sqrt{3}}
    \right)}_{< \, 2}
    +
    \frac{2^{m+\tau+2}}{\lambda_1}
    \underbrace{\sqrt{2}
    \left(
      1
      +
      \frac{7}{2 \sqrt{3}}
    \right)}_{< \, 2^{3}}
    +
    5 \\
    &<
    2^{\Delta+\tau+4} + 2^{\tau+t+5} + 5 = \Nspace
  \end{align*}
  where we have used that $\lambda_1 \ge 2^{m-t}$, that $\lambda_2^{\perp} \ge \sqrt{3} \, \lambda_2 / 2 \ge \sqrt{3} \, \lambda_1 / 2 \ge \sqrt{3} \cdot 2^{m-t-1}$ by Cl.~\ref{claim:lattice-ineq}, and that $\lambda_1 \lambda_2^{\perp} = 2^{m+\ell+\tau} = 2^{2m-\Delta+\tau}$ by Cl.~\ref{claim:lattice-det}.

  By~\cite[Thm.~3]{galbraith-ruprai}, the expected number of group operations for Gaudry--Schost's algorithm~\cite{gaudry-schost}, as generalized and improved by Galbraith and Ruprai~\cite{galbraith-ruprai}, in the idealized model, in the best, average and worst cases, is at most\footnote{We write ``at most'' here since $\Nspace$ is an upper bound on $(2 B_1 + 1)(2 B_2 + 1)$.}
  \begin{align*}
    \left( \left( \frac{2}{\sqrt{3}} \right)^\eta + o(1) \right) \sqrt{\pi \Nspace}
  \end{align*}
  for $\eta$ the dimension, which is two in this case, and so the lemma follows.
\end{proof}

\noindent The enumeration algorithm is described in pseudocode in Alg.~\ref{alg:recover-d-dlp} in App.~\ref{app:algorithms-gaudry-schost}.

Lem.~\ref{lemma:bound-enumeration-dlp} above may be regarded as a variation of~Lem.~\ref{lemma:bound-enumeration}.
It compensates for the ``off-drift'' (see the proof of Lem.~\ref{lemma:bound-enumeration} on p.~\pageref{page:off-drift}) by slightly expanding the search space so as to allow Gaudry--Schost's algorithm (or other algorithms for solving the short multi-dimensional DLP) to be directly called upon.
It thus removes the space barrier in~Lem.~\ref{lemma:bound-enumeration} at the expense of slightly increasing the search space.
On the whole, however, asymptotically as the $o(1)$ term tends to zero, the upper bound on the complexity in Lem.~\ref{lemma:bound-enumeration-dlp} is in fact less than that in~Lem.~\ref{lemma:bound-enumeration}.
It should furthermore be noted that the bounds in Lems.~\ref{lemma:bound-enumeration} and~\ref{lemma:bound-enumeration-dlp} may be tightened, by e.g.\ performing a more detailed analysis, and by leveraging symmetries, that the last component which yields~$d$ is known to be on a restricted interval, and so forth.

Finally, it should be noted that Lem.~\ref{lemma:bound-enumeration-dlp} may be generalized to higher dimensions, and to other quantum algorithms, in a straightforward manner, see App.~\ref{app:algorithms-generalizations}.

\section{Main result}
\label{sect:main-theorem}
We now combine Lems.~\ref{lemma:bound-tau-good-pair}--\ref{lemma:bound-t-balanced-Lj} with Lem.~\ref{lemma:bound-enumeration} and~\ref{lemma:bound-enumeration-dlp} to obtain our main theorems:
\begin{theorem}
  \label{thm:main}
  Let ${\Nspace = 2^{\Delta+\tau+1} + 2^{\tau+t+2} + 2}$, let~$c$ be a positive integer constant, and let~$(j, k)$ be yielded by the quantum algorithm.
  Then, with probability at least
  \begin{align*}
    \max\left(0, 1 - \frac{1}{2^\tau} - \frac{1}{2 \cdot 2^{2\tau}} - \frac{1}{6 \cdot 2^{3\tau}} \right)
    \cdot
    \max\left(0, 1 - 2^{\Delta - 2(t - 1) - \tau} \right)
  \end{align*}
  at most $2^3 c \sqrt{\Nspace}$ group operations in $\langle g \rangle$ have to be performed to recover the logarithm~$d$ from~$(j, k)$ by enumerating vectors in~$\mathcal L^\tau(j)$, for~$m \in \mathbb Z$ an upper bound on the bit length of~$d$, $\ell = m - \Delta$ for $\Delta \in [0, m) \cap \mathbb Z$, $\tau \in [0, \ell] \cap \mathbb Z$ and $t \in [0, m) \cap \mathbb Z$.
  This holds assuming that a few group elements are first pre-computed, and that there is space to store at most~$2^3 \sqrt{\Nspace} / c + 3$ integers in a lookup table.
\end{theorem}
\begin{proof}
  By Lem.~\ref{lemma:bound-t-balanced-Lj}, the integer~$j$ observed is such that~$\mathcal L^\tau(j)$ is not $t$-balanced with probability at most $2^{\Delta - 2(t - 1) - \tau}$.
  By Lem.~\ref{lemma:bound-tau-good-pair}, for any given~$j$, the probability that the integer~$k$ observed given~$j$ is such that~$(j, k)$ is a $\tau$-good pair is at least
  \begin{align*}
    1 - \psi'(2^\tau)
    >
    1 - \frac{1}{2^\tau} - \frac{1}{2 \cdot 2^{2\tau}} - \frac{1}{6 \cdot 2^{3\tau}}.
  \end{align*}

  By Lem.~\ref{lemma:bound-enumeration} at most $2^3 c \sqrt{\Nspace}$ group operations in $\langle g \rangle$ have to be performed to recover~$d$ from~$(j, k)$ by enumerating vectors in~$\mathcal L^\tau(j)$, provided that the two aforementioned conditions are fulfilled, that a few group elements are first pre-computed, and that there is space to store at most~$2^3 \sqrt{\Nspace} \big/ c + 3$ integers in a lookup table, and so the theorem follows.
  Note that we take the maximum of the two lower bounds yielded by Lem.~\ref{lemma:bound-tau-good-pair} and Lem.~\ref{lemma:bound-t-balanced-Lj}, respectively, since both bounds may be negative for certain parameter choices.
\end{proof}

\begin{theorem}[Variation of Thm.~\ref{thm:main}]
  \label{thm:main-dlp}
  Let ${\Nspace = 2^{\Delta+\tau+4} + 2^{\tau+t+5} + 5}$, and let~$(j, k)$ be yielded by the quantum algorithm.
  Then, with probability at least
  \begin{align*}
    \max\left(0, 1 - \frac{1}{2^\tau} - \frac{1}{2 \cdot 2^{2\tau}} - \frac{1}{6 \cdot 2^{3\tau}} \right)
    \cdot
    \max\left(0, 1 - 2^{\Delta - 2(t - 1) - \tau} \right)
  \end{align*}
  the expected number of group operations required to solve~$(j, k)$ for~$d$ by reducing the enumeration problem in~$\mathcal{L}^\tau(j)$ to a two-dimensional short DLP and solving it via Gaudry--Schost's algorithm~\cite{gaudry-schost}, as generalized and improved by Galbraith and Ruprai~\cite{galbraith-ruprai}, in the idealized model, in the best, average and worst cases, is at most $(4/3 + o(1)) \, \sqrt{\pi \Nspace}$.
  This holds assuming that a few group elements are first pre-computed.
\end{theorem}
\begin{proof}
  By Lem.~\ref{lemma:bound-t-balanced-Lj}, the integer~$j$ observed is such that~$\mathcal L^\tau(j)$ is not $t$-balanced with probability at most $2^{\Delta - 2(t - 1) - \tau}$.
  By Lem.~\ref{lemma:bound-tau-good-pair}, for any given~$j$, the probability that the integer~$k$ observed given~$j$ is such that~$(j, k)$ is a $\tau$-good pair is at least
  \begin{align*}
    1 - \psi'(2^\tau)
    >
    1 - \frac{1}{2^\tau} - \frac{1}{2 \cdot 2^{2\tau}} - \frac{1}{6 \cdot 2^{3\tau}}.
  \end{align*}

  By Lem.~\ref{lemma:bound-enumeration-dlp}, the expected number of group operations required to solve~$(j, k)$ for~$d$ by reducing the enumeration problem in~$\mathcal{L}^\tau(j)$ to a two-dimensional short DLP and solving it via Gaudry--Schost's algorithm~\cite{gaudry-schost}, as generalized and improved by Galbraith and Ruprai~\cite{galbraith-ruprai}, in the idealized model, in the best, average and worst cases, is at most $(4/3 + o(1)) \, \sqrt{\pi \Nspace}$, and so the theorem follows.
  Note that we take the maximum of the two lower bounds yielded by Lem.~\ref{lemma:bound-tau-good-pair} and Lem.~\ref{lemma:bound-t-balanced-Lj}, respectively, since both bounds may be negative for certain parameter choices.
\end{proof}

The bounds in Thms.~\ref{thm:main}--\ref{thm:main-dlp} are tabulated in Tabs.~\ref{tab:bounds}--\ref{tab:bounds-continued} in App.~\ref{app:tables} for various~$\Delta$.
More specifically, for~$\Delta$ and a given lower bound on the success probability, the tables give~$t$ and~$\tau$ that minimize the upper bound on the enumeration complexity in Thm.~\ref{thm:main}.
As may be seen in Tabs.~\ref{tab:bounds}--\ref{tab:bounds-continued}, the success probability can easily be pushed as high as $1 - 10^{-10}$ for $\Delta = 0$ when requiring the classical post-processing to be feasible to perform in practice on an ordinary computer.
We can afford to grow~$\Delta$ quite large, depending on which lower bound on the success probability we aim for, and on what amount of computational resources that we are prepared to spend on the post-processing.

\subsection{Notes on our notion of shortness}
\label{sect:notion-shortness}
In Sect.~\ref{sect:quantum-algorithm}, we assumed~$d$ to be short in the sense that $r \ge 2^{m+\ell} + (2^\ell - 1) d$ so as to simplify the analysis by not having to account for modular reductions.

For FF-DH in safe-prime groups with short exponents, the order~$r$ of~$g$ is known, and $d \lll r$, so it trivially holds that $r \ge 2^{m+\ell} + (2^\ell - 1) d$.
In Tab.~\ref{tab:bounds-ff-dh} in App.~\ref{app:tables-ff-dh}, we tabulate the bounds in Thms.~\ref{thm:main}--\ref{thm:main-dlp} for common FF-DH parameterizations.

For RSA, the probability of a random $g \in \mathbb Z_{\Nmodulus}^*$ having order~$r \ge 2^{m+\ell} + (2^\ell - 1) d$ is lower bounded\footnote{Under certain assumptions, see~\cite[App.~A.2.2]{ekera-pp} for the full details.} in~\cite[App.~A.2.2]{ekera-pp}, for~$\Nmodulus$ a large random RSA integer, and shown to be at least $0.9998$ for $\Delta = 20$.
In Tab.~\ref{tab:bounds-rsa} in App.~\ref{app:tables-rsa}, we tabulate the bounds in Thms.~\ref{thm:main}--\ref{thm:main-dlp} for $\Delta = 20$ whilst accounting for this additional reduction factor.
Furthermore, we include a selection of~$\Delta$, and their associated reduction factors, in Tab.~\ref{tab:bounds-rsa}, to reach success probabilities ranging from $\ge 0.9$ to $\ge 1 - 10^{-4}$.

Finally, as explained above, the assumption that $r \ge 2^{m+\ell} + (2^\ell - 1) d$ was made to simplify the analysis:
The assumption may be relaxed, see the heuristic analysis in~\cite{ekera-revisiting} for further details; in particular see \cite[Sect.~7.2 and App.~B.1.2]{ekera-revisiting}.

\subsection{Asymptotic analysis}
Asymptotically, we may select the parameters~$\Delta$, $\tau$ and~$t$ so that the success probability tends to one as~$m$ tends to infinity, whilst preserving the polynomial runtime:
\begin{corollary}
  \label{cor:main}
  The parameters~$\Delta$, $\tau$ and~$t$ may be selected as functions of~$m$ so that the lower bound on the success probability in Thms.~\ref{thm:main}--\ref{thm:main-dlp} tends to one as $m \rightarrow \infty$, whilst the upper bound on the enumeration complexity is $O(\poly(m))$.
\end{corollary}
\begin{proof}
  The corollary follows by e.g.~fixing~$\Delta$ and~$t$ to some constant values, whilst letting $\tau = \log_2 f(m)$ where $f(m) \in \omega_m(1)$ and $f(m) \in O(\poly(m))$.

  Another option is to fix~$t$ to a constant value, whilst letting $\tau = \log_2 f(m)$ for $f(m)$ as above, and $\Delta = \log_2 g(m)$ where $g(m) \in \omega_m(1)$ and $g(m) \in o(f(m))$.
\end{proof}

As is stated in the proof of Cor.~\ref{cor:main}, we may let~$\Delta$ slowly tend to infinity with~$m$.
By the analysis in~\cite[App.~A.2]{ekera-pp}, this implies that the probability of meeting the requirement that $r \ge 2^{m+\ell} + (2^\ell - 1)d$ can be made to tend to one asymptotically when Ekerå--Håstad's algorithm for the short DLP is used to break RSA.

\subsection{Notes on physical implementation}
In this analysis we have assumed that the quantum computer executes the quantum algorithm as per its mathematical description.
If the algorithm is to be executed on a computer that may make computational errors, then the risk of such errors causing the computation to fail\footnote{In the sense that the aforementioned assumption that the quantum computer executes the quantum algorithm as per its mathematical description is void.} must be factored into the success probability.

Furthermore, we describe only logical quantum circuits and consider only logical space and computational costs in this analysis, without accounting for effects and overheads induced by quantum error correction.

\subsection{Notes on practical verification and future work}
We have implemented the post-processing algorithms in Alg.~\ref{alg:recover-d}--\ref{alg:recover-d-dlp} and verified that they work as expected by post-processing simulated quantum algorithm outputs.\footnote{For an accessible but unoptimized implementation of Alg.~\ref{alg:recover-d} and the simulator, see the \href{https://github.com/ekera/quaspy}{repository for Quaspy}~\cite{quaspy} on GitHub, available at~\url{https://github.com/ekera/quaspy}.}

Our initial experiments with an optimized parallelized implementation indicate that it is typically not an issue to run the post-processing on an ordinary computer for $\Delta = 50$ when targeting a $\ge 99\%$ success probability.
To exemplify, this leads to a reduction by approximately~$15\%$ in the number of group operations that need to be performed quantumly to solve the short DLP in 2048-bit safe-prime groups in a single run compared to the baseline costs for $\Delta = 0$ reported in~\cite[Tab.~2]{ekera-pp}.\footnote{To be specific, the reduction is from $672$~operations as reported in~\cite[Tab.~2]{ekera-pp} to $572$~operations.}
Work is currently underway\footnote{This work will form part of an MSc thesis supervised by the author.} to optimize and further parallelize said implementation.

\subsection{Notes on generalizations and future work}
As previously stated, a key to achieving the results in this paper is to efficiently perform a limited search in the classical post-processing by leveraging meet-in-the-middle or random-walk techniques.
These techniques may be generalized to speed up other related classical post-processing algorithms that perform limited searches, such as those in~\cite{ekera-pp, ekera-general, ekera-success, ekera-revisiting}, both when making and not making tradeoffs, see App.~\ref{app:algorithms-generalizations}.

\section*{Acknowledgements}
I am grateful to Johan Håstad for valuable comments and advice.
I thank Joel Gärtner for identifying an issue in Lem.~\ref{lemma:bound-enumeration} in the initial pre-print version of this manuscript.
Funding and support for this work was provided by the Swedish NCSA that is a part of the Swedish Armed Forces.
Computations were enabled by resources provided by the National Academic Infrastructure for Supercomputing in Sweden (NAISS) at PDC at KTH partially funded by the Swedish Research Council through grant agreement no.~2022-06725.

\clearpage

\appendix
\section{Algorithms}
\label{app:algorithms}
In this appendix, we describe the post-processing algorithms in pseudocode.

\subsection{Solving via a generalization of Shanks' algorithm}
\label{app:algorithms-shanks}
\begin{breakablealgorithm}
  \caption{Returns~$d$ given~$g$, $x = g^d$, a $\tau$-good pair~$(j, k)$, $c \in \mathbb Z_{> 0}$, $m$ and~$\ell$.}
  \label{alg:recover-d}

  \begin{pseudocode}
  \item Let $\textsc{MeetInTheMiddle}(g, x, \nu_1, \nu_2, B_1, B_2, s_1, s_2, \mu, c)$ be the function:

  \begin{pseudocode}
      \item Let $g_1 \leftarrow g^{s_1}$, $g_2 \leftarrow g^{s_2}$ and $w = g_1^{\nu_1} \cdot g_2^{\nu_2} \cdot x^{-1}$.

      \vspace{2mm}

      \item Let $n \leftarrow c \, \round{\sqrt{B_1 / (B_2 + 1)}}$.

      \emph{Note: As $B_1 \ge 1$ and $2 B_1 > B_2 \ge 0$, see Cl.~\ref{claim:bounds-B1-B2}, it holds that $n \ge c \ge 1$.}

      \vspace{2.5mm}

      \item Let~$T$ be an empty lookup table. \hfill \emph{--- Note: The first stage begins.}

      Insert~$0$ into~$T$ indexed by~$g^0$.

      \vspace{1mm}

      \item Let $s \leftarrow g_1^{n} = g^{n \cdot s_1}$, $z_{+} \leftarrow s$, $z_{-} \leftarrow s^{-1}$ and $i \leftarrow 1$.

      \vspace{0.5mm}

      Store~$s^{-1}$ in memory as a pre-computed group element.

      \vspace{1mm}

      \item Repeat: \label{alg:recover-d:step:begin-first-stage}

      \begin{pseudocode}
        \item Insert~$i$ into~$T$ indexed by $z_{+} = g^{i \cdot n \cdot s_1}$.

              Insert~$-i$ into~$T$ indexed by $z_{-} = g^{-i \cdot n \cdot s_1}$.

              \emph{Note: This step is visited $\ceil{B_1 / n}$ times.}

        \item Let $i \leftarrow i + 1$. If $i > \ceil{B_1 / n}$:
        \begin{pseudocode}
          \item Stop repeating and go to step~\ref{alg:recover-d:step:end-first-stage}.
        \end{pseudocode}

        \item Let $z_{+} \leftarrow z_{+} \cdot s$ and $z_{-} \leftarrow z_{-} \cdot s^{-1}$.

        \emph{Note: This step is visited $\ceil{B_1 / n} - 1$ times.}
      \end{pseudocode}

      \vspace{1mm}

      \item Let $z_+ \leftarrow w$, $z_- \leftarrow w$ and $j \leftarrow 0$. \label{alg:recover-d:step:end-first-stage} \hfill \emph{--- Note: The second stage begins.}

      \vspace{1mm}

      Store~$g_1^{-1}$ and~$g_2^{-1}$ in memory as pre-computed group elements.
      Also pre-compute and store ${g_2 \cdot g_1}$, ${g_2 \cdot g_1^{-1}}$, ${g_2^{-1} \cdot g_1}$ and ${g_2^{-1} \cdot g_1^{-1}}$.

      \vspace{1mm}

      \item Repeat: \label{alg:recover-d:step:begin-second-stage} 
      \begin{pseudocode}
        \item Let $z'_{+} \leftarrow z_{+}$, $z'_{-} \leftarrow z_{-}$ and $i \leftarrow 0$.

        \vspace{2mm}

        \item Repeat: 

        \emph{Note: At this point $z'_{\pm} = g^{(\nu_1 + i - \round{\pm j \cdot \mu}) \cdot s_1 + (\nu_2 \pm j) \cdot s_2} \cdot x^{-1}$.}

        \begin{pseudocode}
          \item If $z'_{+}$ indexes an integer~$k$ in~$T$:

          \emph{Note:}
          If this is the case, then $z'_{+} = g^{k \cdot n \cdot s_1}$.

          \begin{pseudocode}
            \item Return $d = (\nu_1 + i - \round{\, j \cdot \mu} - k \cdot n) \cdot s_1 + (\nu_2 + j) \cdot s_2$.
          \end{pseudocode}

          \item If $j > 0$ and~$z'_{-}$ indexes an integer~$k$ in~$T$:

          \emph{Note:}
          If this is the case, then $z'_{-} = g^{k \cdot n \cdot s_1}$.

          \begin{pseudocode}
            \item Return $d = (\nu_1 + i - \round{-j \cdot \mu} - k \cdot n) \cdot s_1 + (\nu_2 - j) \cdot s_2$.
          \end{pseudocode}

          \item Let $i \leftarrow i + 1$. If $i \ge n$:
          \begin{pseudocode}
            \item Stop repeating and go to step~\ref{alg:recover-d:step:end-second-stage-inner}.
          \end{pseudocode}

          \item Let $z'_{+} \leftarrow z'_{+} \cdot g_1$ and $z'_{-} \leftarrow z'_{-} \cdot g_1$.

          \emph{Note: This step is visited $(B_2 + 1)(n - 1)$ times.}
        \end{pseudocode}

        \item Let $j \leftarrow j + 1$. If $j > B_2$: \label{alg:recover-d:step:end-second-stage-inner}
        \begin{pseudocode}
          \item Stop repeating and go to step~\ref{alg:recover-d:step:end-second-stage-outer}.
        \end{pseudocode}

        \item Update $z_+$ and $z_-$ as follows:

        \begin{pseudocode}
          \item If $\round{\, j \cdot \mu} < \round{(j - 1) \cdot \mu}$:

          \begin{pseudocode}
            \item Let $z_+ \leftarrow z_+ \cdot g_2 \cdot g_1$ and $z_- \leftarrow z_- \cdot g_2^{-1} \cdot g_1^{-1}$.
          \end{pseudocode}

          \item Otherwise, if $\round{\, j \cdot \mu} > \round{(j - 1) \cdot \mu}$:

          \begin{pseudocode}
            \item Let $z_+ \leftarrow z_+ \cdot g_2 \cdot g_1^{-1}$ and $z_- \leftarrow z_- \cdot g_2^{-1} \cdot g_1$.
          \end{pseudocode}

          \item Otherwise:

          \begin{pseudocode}
            \item Let $z_+ \leftarrow z_+ \cdot g_2$ and $z_- \leftarrow z_- \cdot g_2^{-1}$.
          \end{pseudocode}
        \end{pseudocode}

        \emph{Note: This step is visited~$B_2$ times.
        The group elements used above to update~$z_+$ and~$z_-$, respectively, are all pre-computed.
        Also, since ${|\, \mu \,| \le 1/2}$, it holds that ${\round{\, j \cdot \mu} - \round{(j - 1) \cdot \mu} \in \{ -1, 0, 1 \}}$.}
      \end{pseudocode}

      \vspace{1mm}

      \item Return~$\neg$. \label{alg:recover-d:step:end-second-stage-outer}
    \end{pseudocode}

    \vspace{1mm}

    \item Let~$\mathcal L^\tau(j)$ be the lattice generated by~$(j, 2^\tau)$ and~$(2^{m+\ell}, 0)$.

    \vspace{2mm}

    \item Let~$\vec s_1 = (s_{1,1}, s_{1,2})$ of norm~$\lambda_1$ be a shortest non-zero vector in~$\mathcal L^\tau(j)$, and let ${\vec s_2 = (s_{2,1}, s_{2,2})}$ be a shortest non-zero vector in~$\mathcal L^\tau(j)$ linearly independent to~$\vec s_1$, so that~$(\vec s_1, \vec s_2)$ forms a Lagrange-reduced basis.

    \emph{Note: The basis $(\vec s_1, \vec s_2)$ may be found with Lagrange's algorithm, see~\cite{lagrange, nguyen}.}

    \vspace{2mm}

    \item Let $\mu = \langle \vec s_1, \vec s_2 \rangle / \lambda_1^2$.
    Let~$\vec s_2^\parallel = \mu \cdot \vec s_1$ be the component of~$\vec s_2$ parallel to~$\vec s_1$, and let~$\vec s_2^\perp = \vec s_2 - \vec s_2^\parallel$ of norm~$\lambda_2^\perp$ be the component of~$\vec s_2$ orthogonal to~$\vec s_1$.

    \emph{Note: As~$(\vec s_1$, $\vec s_2)$ is Lagrange-reduced, it holds that $|\, \mu \,| \le 1/2$, see Cl.~\ref{claim:lattice-ineq}.}

    \vspace{2mm}

    \item Let $\vec v = (\{-2^m k\}_{2^{m+\ell}}, 0) \in \mathbb Z^2$, and let~$\vec o$ be the vector in~$\mathcal L^\tau(j)$ yielded by Babai's nearest plane algorithm~\cite{babai} upon input of~$\vec v$ and the basis $(\vec s_1, \vec s_2)$.

    Let~$\nu_1$ and~$\nu_2$ be integers such that $\vec o = \nu_1 \vec s_1 + \nu_2 \vec s_2$.

    \vspace{2mm}

    \item Let $B_1 \leftarrow \floor{2^{m+\tau} \sqrt{2} / \lambda_1 + 1}$ and $B_2 \leftarrow \floor{2^{m+\tau} \sqrt{2} / \lambda_2^{\perp} + 1/2}$.

    \emph{Note: It holds that $B_1 \ge 1$ and $2 B_1 > B_2 \ge 0$, see Cl.~\ref{claim:bounds-B1-B2}.}

    \vspace{2mm}

    \item Return $\textsc{MeetInTheMiddle}(g, x, \nu_1, \nu_2, B_1, B_2, s_{1, 2} / 2^\tau, s_{2, 2} / 2^\tau, \mu, c)$.
  \end{pseudocode}
\end{breakablealgorithm}

\noindent Note that the fact that Alg.~\ref{alg:recover-d} requires four inverses\footnote{More specifically~$g_1^{-1}$, $g_2^{-1}$, $s^{-1}$ and~$x^{-1}$.} to be computed does not imply a loss of generality in the context of this work since the quantum algorithm in Sect.~\ref{sect:quantum-algorithm} also requires inverses to be computed.
It may furthermore be necessary to compute inverses to implement the group arithmetic reversibly quantumly, see Sect.~\ref{sect:quantum-algorithm-implementation} for further details.

\subsubsection{Notes on handling the ``off-drift''}
As stated in the proof of Lem.~\ref{lemma:bound-enumeration} in Sect.~\ref{sect:bound-enumeration-complexity-shanks}, the ``off-drift'' in the direction of~$\vec s_1$ when adding ${m_2 \vec s_2}$ to~$\vec o$ is compensated for by at the same time subtracting ${\round{m_2 \cdot \mu} \, \vec s_1}$ in Alg.~\ref{alg:recover-d}.

Another option is to increase the enumeration bounds $B_1, \, B_2, \, \ldots$ so as to ensure that all lattice vectors within the prescribed radius are included in the enumeration even when not compensating for the ``off-drift'' by subtracting.
This option is used in Alg.~\ref{alg:recover-d-dlp}, where it gives rise to a short multi-dimensional DLP that can be solved using standard algorithms.

\subsubsection{Notes on parallelization and space-saving optimizations}
As explained in Sect.~\ref{sect:bound-enumeration-complexity-gaudry-schost}, the space usage is typically a limiting factor when using Alg.~\ref{alg:recover-d} and attempting to select large~$\Delta$ and/or large~$t$ and~$\tau$.
To completely remove this space barrier, a good option is to forego using Alg.~\ref{alg:recover-d} and deterministic meet-in-the-middle techniques altogether, and to instead proceed via Alg.~\ref{alg:recover-d-dlp} (see the next section) that reduces the lattice enumeration problem to a short multi-dimensional DLP that can be solved probabilistically via Gaudry--Schost's algorithm~\cite{gaudry-schost}, as generalized and improved by Galbraith and Ruprai~\cite{galbraith-ruprai}.
This reduces the space usage to~$O(1)$ group elements, in analogy with how Pollard~\cite{pollard-rho-lambda, oorschot-wiener} randomized Shanks' algorithm~\cite{shanks} in the one-dimensional case so as to avoid having to store more than~$O(1)$ group elements.\footnote{This idea is also discussed in the ``Notes on randomization'' paragraph on p.~85 of~\cite{ekera-phd-thesis}.}

Another option for reducing the space usage in Alg.~\ref{alg:recover-d} itself is to replace the lookup table~$T$ with a Bloom filter~\cite{bloom} or some more modern filter that tests set membership.
At the expense of performing some more computational work, this may for instance be accomplished as follows:

In the first stage, the elements $g^{k \cdot n \cdot s_1}$ are inserted into the filter~$F$.
In the second stage, the indices~$(i, j)$ of the elements $z'_{\pm} = g^{(\nu_1 + i - \round{\pm j \cdot \mu}) \cdot s_1 + (\nu_2 \pm j) \cdot s_2} \cdot x^{-1}$ found to be in~$F$ are inserted into a small lookup table~$T'$ indexed by $z'_{\pm}$.
The first stage is then re-executed and the elements $g^{k \cdot n \cdot s_1}$ looked up in~$T'$ to find a tuple $(i, j, k)$ such that {$d = (\nu_1 + i - \round{\, j \cdot \mu} - k \cdot n) \cdot s_{1} + (\nu_2 + j) \cdot s_{2}$}.

Yet another option for managing the space usage is to distribute the lookup table~$T$ (across nodes, for instance, or by offloading~$T$ to a (distributed) file system, or similar), and to use a filter that tests set membership to filter the lookups so that only lookups of elements that are actually likely to be stored in~$T$ are performed (so as to avoid performing too many slow lookups in~$T$).

Finally, note that for ease of comprehension and analysis, Alg.~\ref{alg:recover-d} is described in a simple sequential manner.
If the overall runtime is a limiting factor then the main loops in the two stages of Alg.~\ref{alg:recover-d} (i.e.\ the loops in steps~\ref{alg:recover-d:step:begin-first-stage} and~\ref{alg:recover-d:step:begin-second-stage}) may be parallelized (provided that~$T$ or~$F$ is implemented in a manner that admits parallelization) at the expense of performing some more pre-computational work.

\subsection{Solving via Gaudry--Schost's algorithm}
\label{app:algorithms-gaudry-schost}
\begin{breakablealgorithm}
  \caption{Returns~$d$ given~$g$, $x = g^d$, a $\tau$-good pair~$(j, k)$, $m$ and~$\ell$.}
  \label{alg:recover-d-dlp}

  \begin{pseudocode}
    \item Let~$\mathcal L^\tau(j)$ be the lattice generated by~$(j, 2^\tau)$ and~$(2^{m+\ell}, 0)$.

    \vspace{2mm}

    \item Let~$\vec s_1 = (s_{1,1}, s_{1,2})$ of norm~$\lambda_1$ be a shortest non-zero vector in~$\mathcal L^\tau(j)$, and let ${\vec s_2 = (s_{2,1}, s_{2,2})}$ be a shortest non-zero vector in~$\mathcal L^\tau(j)$ linearly independent to~$\vec s_1$, so that~$(\vec s_1, \vec s_2)$ forms a Lagrange-reduced basis.

    \emph{Note: The basis $(\vec s_1, \vec s_2)$ may be found with Lagrange's algorithm, see~\cite{lagrange, nguyen}.}

    \vspace{2mm}

    \item Let $\mu = \langle \vec s_1, \vec s_2 \rangle / \lambda_1^2$.
    Let~$\vec s_2^\parallel = \mu \cdot \vec s_1$ be the component of~$\vec s_2$ parallel to~$\vec s_1$, and let~$\vec s_2^\perp = \vec s_2 - \vec s_2^\parallel$ of norm~$\lambda_2^\perp$ be the component of~$\vec s_2$ orthogonal to~$\vec s_1$.

    \emph{Note: As~$(\vec s_1$, $\vec s_2)$ is Lagrange-reduced, it holds that $|\, \mu \,| \le 1/2$, see Cl.~\ref{claim:lattice-ineq}.}

    \vspace{2mm}

    \item Let $\vec v = (\{-2^m k\}_{2^{m+\ell}}, 0) \in \mathbb Z^2$, and let~$\vec o$ be the vector in~$\mathcal L^\tau(j)$ yielded by Babai's nearest plane algorithm~\cite{babai} upon input of~$\vec v$ and the basis $(\vec s_1, \vec s_2)$.

    Let~$\nu_1$ and~$\nu_2$ be integers such that $\vec o = \nu_1 \vec s_1 + \nu_2 \vec s_2$.

    \vspace{2mm}

    \item Let $B_2 \leftarrow \floor{2^{m+\tau} \sqrt{2} / \lambda_2^{\perp} + 1/2}$ and $B_1 \leftarrow \floor{2^{m+\tau} \sqrt{2} / \lambda_1 + |\, \mu \,| \cdot B_2 + 1/2}$.

    \vspace{2mm}

    \item Let $s_1 \leftarrow s_{1, 2} / 2^\tau$,
              $s_2 \leftarrow s_{2, 2} / 2^\tau$,
              $g_1 \leftarrow g^{s_1}$,
              $g_2 \leftarrow g^{s_2}$ and
              $x' \leftarrow x \, g_1^{-\nu_1} \, g_2^{-\nu_2}$.

    \vspace{2mm}

    \item Let $(i_1, i_2) \leftarrow \textsc{SolveTwodimensionalShortDLP}(x', g_1, g_2, B_1, B_2)$.

    \emph{Note: Solves $g_1^{i_1} \, g_2^{i_2} = x'$ for $(i_1, i_2)$ where
    \begin{align*}
      i_1 \in [-B_1, B_1] \cap \mathbb Z
      \quad \text{ and } \quad
      i_2 \in [-B_2, B_2] \cap \mathbb Z
    \end{align*}
    and returns $(i_1, i_2)$, or $(\neg, \neg)$ if no solution is found.
    The function called here may e.g.\ be implemented with Gaudry--Schost's algorithm~\cite{gaudry-schost, galbraith-ruprai}.}

    \vspace{2mm}

    \item If $(i_1, i_2) = (\neg, \neg)$:
    \begin{pseudocode}
      \item Return $\neg$.
    \end{pseudocode}

    \item Return $d = (\nu_1 + i_1) s_1 + (\nu_2 + i_2) s_2$.
  \end{pseudocode}
\end{breakablealgorithm}

\noindent Note that the fact that Alg.~\ref{alg:recover-d-dlp} requires two inverses\footnote{More specifically~$g_1^{-1}$ and $g_2^{-1}$.} to be computed does not imply a loss of generality in the context of this work since the quantum algorithm in Sect.~\ref{sect:quantum-algorithm} also requires inverses to be computed.
It may furthermore be necessary to compute inverses to implement the group arithmetic reversibly quantumly, see Sect.~\ref{sect:quantum-algorithm-implementation} for further details.

\subsubsection{Notes on parallelization}
Gaudry--Schost's algorithm~\cite{gaudry-schost} (with Galbraith--Ruprai's improvements~\cite{galbraith-ruprai}) can be trivially parallelized with very small communication and storage overheads.

\subsection{Notes on generalizations to related quantum algorithms}
\label{app:algorithms-generalizations}
The techniques used in Algs.~\ref{alg:recover-d}--\ref{alg:recover-d-dlp} may be leveraged to speed up related classical post-processing algorithms that perform limited searches, such as the lattice-based algorithms in~\cite{ekera-pp, ekera-general, ekera-success, ekera-revisiting}, both when solving in a single run and when making tradeoffs between the number of runs of the quantum algorithm and the number of group operations that need to be evaluated quantumly in each run.

To provide some more details, the above referenced classical post-processing algorithms perform an enumeration in a lattice and test whether the last component of each vector enumerated fulfills a requirement by exponentiating a group element to the value of the component multiplied by a scaling factor.
This is exactly the situation in Algs.~\ref{alg:recover-d}--\ref{alg:recover-d-dlp}.
Essentially only the lattice basis, the vectors~$\vec u$ and~$\vec v$, the enumeration bounds, and the scaling factor, must be adapted in Algs.~\ref{alg:recover-d}--\ref{alg:recover-d-dlp} to make them cover the above cases.

Furthermore, in the case of solving an order-finding problem (OFP) as opposed to a DLP, a number of candidate solutions that meet the test will typically need to be returned, and not only the first such candidate.
The trivial solution must furthermore be ignored.
This requires a straightforward adaptation of Alg.~\ref{alg:recover-d}, and an adaptation of Gaudry--Schost's algorithm~\cite{gaudry-schost, galbraith-ruprai} that is called by Alg.~\ref{alg:recover-d-dlp} (to make it find collisions between so-called ``tame'' walks, as opposed to between ``tame'' and ``wild'' walks).

\paragraph{On tradeoffs}
When making tradeoffs and solving in multiple runs, the lattice dimension~$\eta$ increases above two.
A straightforward way to handle this is by proceeding as in Lem.~\ref{lemma:bound-enumeration-dlp} and Alg.~\ref{alg:recover-d-dlp}, whilst replacing Lagrange's algorithm~\cite{lagrange} with the LLL algorithm~\cite{lll}, and deriving independent\footnote{In the sense that the bounds are independent of the enumeration indices, not of each other.} enumeration bounds from the Gram--Schmidt orthogonalization of the LLL-reduced basis.
For $i \in [1, \eta] \cap \mathbb Z$, the enumeration bounds then become
\begin{align*}
  B_i
  \le
  \floor{\frac{R}{\lambda^*_i} + \sum_{j \, = \, i + 1}^{\eta} |\, \mu_{j, i} \,| \cdot B_j + \frac{1}{2}}
  \quad \text{ where } \quad
  \mu_{j,i}
  =
  \frac{\langle \vec s_j, \vec s_i^* \rangle}{(\lambda_i^*)^2}
  \quad \text{ and } \quad
  |\, \mu_{j,i} \,|
  \le
  \frac{1}{2}
\end{align*}
for~$R$ the enumeration radius, and $\lambda_1^*, \, \ldots, \, \lambda_\eta^*$ the norms of the vectors $\vec s_1^*, \, \ldots, \, \vec s_\eta^*$ in the Gram--Schmidt orthogonalization of the vectors $\vec s_1, \, \ldots, \, \vec s_\eta$ in the LLL-reduced basis.
(Note that these bounds are as in Lem.~\ref{lemma:bound-enumeration-dlp} and Alg.~\ref{alg:recover-d-dlp} when $\eta = 2$ since $\lambda_2^\perp = \lambda_2^*$ and $\mu_{2,1} = \mu$.)
This yields a multi-dimensional short DLP that may be solved by Gaudry--Schost's algorithm~\cite{gaudry-schost} as generalized and improved by Galbraith and Ruprai~\cite{galbraith-ruprai}.
Alternatively, the multi-dimensional short DLP may be solved by generalizing Shanks' algorithm~\cite{shanks}, by dividing the vectors to be enumerated into two sets of approximately equal size.

Finally, note that when making tradeoffs for large tradeoff factors, one would typically pick the number of runs so that the number of vectors to be enumerated is small.
The benefit of using meet-in-the-middle or random-walk techniques is then fairly limited.
This explains why we focus primarily on the single-run setting in this work.
Small tradeoff factors requiring only a small number of runs are also of interest, however.
The results in this work may be extended to such multiple-run settings as explained above.

\clearpage

\section{Tables}
\label{app:tables}
In this appendix, we tabulate the lower bound on the success probability, and the associated upper bound on the enumeration complexity, in Thm.~\ref{thm:main}, in~$\Delta$, $\tau$ and~$t$:

\begin{table}[h!]
  \begin{center}
    \begin{tabular}{r|rr|l|r}
      \hline
               &        &     & Success     & Work \\
      $\Delta$ & $\tau$ & $t$ & probability & ($\log_2$) \\
      \thickhline
       0 &  4 &  2 & $\ge 0.9$          & $\le  7.1$ \\ 
         &  5 &  2 & $\ge 0.95$         & $\le  7.6$ \\ 
         &  7 &  2 & $\ge 0.99$         & $\le  8.6$ \\ 
         & 11 &  1 & $\ge 0.999$        & $\le 10.2$ \\ 
         & 14 &  2 & $\ge 1 - 10^{-4}$  & $\le 12.1$ \\ 
         & 17 &  2 & $\ge 1 - 10^{-5}$  & $\le 13.6$ \\ 
         & 21 &  1 & $\ge 1 - 10^{-6}$  & $\le 15.2$ \\ 
         & 24 &  2 & $\ge 1 - 10^{-7}$  & $\le 17.1$ \\ 
         & 27 &  2 & $\ge 1 - 10^{-8}$  & $\le 18.6$ \\ 
         & 31 &  1 & $\ge 1 - 10^{-9}$  & $\le 20.2$ \\ 
         & 34 &  2 & $\ge 1 - 10^{-10}$ & $\le 22.1$ \\ 
      \hline
      10 &  4 &  7 & $\ge 0.9$          & $\le 10.7$ \\ 
         &  5 &  7 & $\ge 0.95$         & $\le 11.2$ \\ 
         &  7 &  7 & $\ge 0.99$         & $\le 12.2$ \\ 
         & 10 &  9 & $\ge 0.999$        & $\le 14.1$ \\ 
         & 14 &  7 & $\ge 1 - 10^{-4}$  & $\le 15.7$ \\ 
         & 17 &  7 & $\ge 1 - 10^{-5}$  & $\le 17.2$ \\ 
         & 20 &  9 & $\ge 1 - 10^{-6}$  & $\le 19.1$ \\ 
         & 24 &  7 & $\ge 1 - 10^{-7}$  & $\le 20.7$ \\ 
         & 27 &  7 & $\ge 1 - 10^{-8}$  & $\le 22.2$ \\ 
         & 30 &  8 & $\ge 1 - 10^{-9}$  & $\le 23.8$ \\ 
         & 34 &  7 & $\ge 1 - 10^{-10}$ & $\le 25.7$ \\ 
      \hline
      20 &  4 & 12 & $\ge 0.9$          & $\le 15.6$ \\ 
         &  5 & 12 & $\ge 0.95$         & $\le 16.1$ \\ 
         &  7 & 12 & $\ge 0.99$         & $\le 17.1$ \\ 
         & 10 & 14 & $\ge 0.999$        & $\le 18.6$ \\ 
         & 14 & 12 & $\ge 1 - 10^{-4}$  & $\le 20.6$ \\ 
         & 17 & 12 & $\ge 1 - 10^{-5}$  & $\le 22.1$ \\ 
         & 20 & 14 & $\ge 1 - 10^{-6}$  & $\le 23.6$ \\ 
         & 24 & 12 & $\ge 1 - 10^{-7}$  & $\le 25.6$ \\ 
         & 27 & 12 & $\ge 1 - 10^{-8}$  & $\le 27.1$ \\ 
         & 30 & 13 & $\ge 1 - 10^{-9}$  & $\le 28.6$ \\ 
         & 34 & 12 & $\ge 1 - 10^{-10}$ & $\le 30.6$ \\ 
      \hline
    \end{tabular}
    $\:$
    \begin{tabular}{r|rr|l|r}
      \hline
               &        &     & Success     & Work \\
      $\Delta$ & $\tau$ & $t$ & probability & ($\log_2$) \\
      \thickhline
      30 &  4 & 17 & $\ge 0.9$          & $\le 20.6$ \\ 
         &  5 & 17 & $\ge 0.95$         & $\le 21.1$ \\ 
         &  7 & 17 & $\ge 0.99$         & $\le 22.1$ \\ 
         & 10 & 19 & $\ge 0.999$        & $\le 23.6$ \\ 
         & 14 & 17 & $\ge 1 - 10^{-4}$  & $\le 25.6$ \\ 
         & 17 & 17 & $\ge 1 - 10^{-5}$  & $\le 27.1$ \\ 
         & 20 & 19 & $\ge 1 - 10^{-6}$  & $\le 28.6$ \\ 
         & 24 & 17 & $\ge 1 - 10^{-7}$  & $\le 30.6$ \\ 
         & 27 & 17 & $\ge 1 - 10^{-8}$  & $\le 32.1$ \\ 
         & 30 & 18 & $\ge 1 - 10^{-9}$  & $\le 33.6$ \\ 
         & 34 & 17 & $\ge 1 - 10^{-10}$ & $\le 35.6$ \\ 
      \hline
      40 &  4 & 22 & $\ge 0.9$          & $\le 25.6$ \\ 
         &  5 & 22 & $\ge 0.95$         & $\le 26.1$ \\ 
         &  7 & 22 & $\ge 0.99$         & $\le 27.1$ \\ 
         & 10 & 24 & $\ge 0.999$        & $\le 28.6$ \\ 
         & 14 & 22 & $\ge 1 - 10^{-4}$  & $\le 30.6$ \\ 
         & 17 & 22 & $\ge 1 - 10^{-5}$  & $\le 32.1$ \\ 
         & 20 & 24 & $\ge 1 - 10^{-6}$  & $\le 33.6$ \\ 
         & 24 & 22 & $\ge 1 - 10^{-7}$  & $\le 35.6$ \\ 
         & 27 & 22 & $\ge 1 - 10^{-8}$  & $\le 37.1$ \\ 
         & 30 & 23 & $\ge 1 - 10^{-9}$  & $\le 38.6$ \\ 
         & 34 & 22 & $\ge 1 - 10^{-10}$ & $\le 40.6$ \\ 
      \hline
      50 &  4 & 27 & $\ge 0.9$          & $\le 30.6$ \\ 
         &  5 & 27 & $\ge 0.95$         & $\le 31.1$ \\ 
         &  7 & 27 & $\ge 0.99$         & $\le 32.1$ \\ 
         & 10 & 29 & $\ge 0.999$        & $\le 33.6$ \\ 
         & 14 & 27 & $\ge 1 - 10^{-4}$  & $\le 35.6$ \\ 
         & 17 & 27 & $\ge 1 - 10^{-5}$  & $\le 37.1$ \\ 
         & 20 & 29 & $\ge 1 - 10^{-6}$  & $\le 38.6$ \\ 
         & 24 & 27 & $\ge 1 - 10^{-7}$  & $\le 40.6$ \\ 
         & 27 & 27 & $\ge 1 - 10^{-8}$  & $\le 42.1$ \\ 
         & 30 & 28 & $\ge 1 - 10^{-9}$  & $\le 43.6$ \\ 
         & 34 & 27 & $\ge 1 - 10^{-10}$ & $\le 45.6$ \\ 
      \hline
    \end{tabular}
  \end{center}
  \caption{The lower bound on the success probability and associated upper bound on the enumeration complexity in Thm.~\ref{thm:main} tabulated in~$\Delta$, $\tau$ and~$t$ for $c = 1$.
  For~$\Delta$ and a given lower bound on the success probability, the table gives~$t$ and~$\tau$ that minimize the enumeration complexity in group operations as given by $\log_2 \sqrt{\Nspace} + 3$ for~$\Nspace$ as in Thm.~\ref{thm:main}.
  The enumeration complexity is reported in the ``Work'' column.
  The bound in said column is also a bound on the enumeration complexity as given by $\log_2(\frac{4}{3} \sqrt{\pi \Nspace})$ for~$\Nspace$ as in Thm.~\ref{thm:main-dlp}.}
  \label{tab:bounds}
\end{table}

\clearpage

\begin{table}[h!]
  \begin{center}
    \begin{tabular}{r|rr|l|r}
      \hline
               &        &     & Success     & Work \\
      $\Delta$ & $\tau$ & $t$ & probability & ($\log_2$) \\
      \thickhline
      60 &  4 & 32 & $\ge 0.9$          & $\le 35.6$ \\ 
         &  5 & 32 & $\ge 0.95$         & $\le 36.1$ \\ 
         &  7 & 32 & $\ge 0.99$         & $\le 37.1$ \\ 
         & 10 & 34 & $\ge 0.999$        & $\le 38.6$ \\ 
         & 14 & 32 & $\ge 1 - 10^{-4}$  & $\le 40.6$ \\ 
         & 17 & 32 & $\ge 1 - 10^{-5}$  & $\le 42.1$ \\ 
         & 20 & 34 & $\ge 1 - 10^{-6}$  & $\le 43.6$ \\ 
         & 24 & 32 & $\ge 1 - 10^{-7}$  & $\le 45.6$ \\ 
         & 27 & 32 & $\ge 1 - 10^{-8}$  & $\le 47.1$ \\ 
         & 30 & 33 & $\ge 1 - 10^{-9}$  & $\le 48.6$ \\ 
         & 34 & 32 & $\ge 1 - 10^{-10}$ & $\le 50.6$ \\ 
      \hline
      70 &  4 & 37 & $\ge 0.9$          & $\le 40.6$ \\ 
         &  5 & 37 & $\ge 0.95$         & $\le 41.1$ \\ 
         &  7 & 37 & $\ge 0.99$         & $\le 42.1$ \\ 
         & 10 & 39 & $\ge 0.999$        & $\le 43.6$ \\ 
         & 14 & 37 & $\ge 1 - 10^{-4}$  & $\le 45.6$ \\ 
         & 17 & 37 & $\ge 1 - 10^{-5}$  & $\le 47.1$ \\ 
         & 20 & 39 & $\ge 1 - 10^{-6}$  & $\le 48.6$ \\ 
         & 24 & 37 & $\ge 1 - 10^{-7}$  & $\le 50.6$ \\ 
         & 27 & 37 & $\ge 1 - 10^{-8}$  & $\le 52.1$ \\ 
         & 30 & 38 & $\ge 1 - 10^{-9}$  & $\le 53.6$ \\ 
         & 34 & 37 & $\ge 1 - 10^{-10}$ & $\le 55.6$ \\ 
      \hline
      80 &  4 & 42 & $\ge 0.9$          & $\le 45.6$ \\ 
         &  5 & 42 & $\ge 0.95$         & $\le 46.1$ \\ 
         &  7 & 42 & $\ge 0.99$         & $\le 47.1$ \\ 
         & 10 & 44 & $\ge 0.999$        & $\le 48.6$ \\ 
         & 14 & 42 & $\ge 1 - 10^{-4}$  & $\le 50.6$ \\ 
         & 17 & 42 & $\ge 1 - 10^{-5}$  & $\le 52.1$ \\ 
         & 20 & 44 & $\ge 1 - 10^{-6}$  & $\le 53.6$ \\ 
         & 24 & 42 & $\ge 1 - 10^{-7}$  & $\le 55.6$ \\ 
         & 27 & 42 & $\ge 1 - 10^{-8}$  & $\le 57.1$ \\ 
         & 30 & 43 & $\ge 1 - 10^{-9}$  & $\le 58.6$ \\ 
         & 34 & 42 & $\ge 1 - 10^{-10}$ & $\le 60.6$ \\ 
      \hline
      90 &  4 & 47 & $\ge 0.9$          & $\le 50.6$ \\ 
         &  5 & 47 & $\ge 0.95$         & $\le 51.1$ \\ 
         &  7 & 47 & $\ge 0.99$         & $\le 52.1$ \\ 
         & 10 & 49 & $\ge 0.999$        & $\le 53.6$ \\ 
         & 14 & 47 & $\ge 1 - 10^{-4}$  & $\le 55.6$ \\ 
         & 17 & 47 & $\ge 1 - 10^{-5}$  & $\le 57.1$ \\ 
         & 20 & 49 & $\ge 1 - 10^{-6}$  & $\le 58.6$ \\ 
         & 24 & 47 & $\ge 1 - 10^{-7}$  & $\le 60.6$ \\ 
         & 27 & 47 & $\ge 1 - 10^{-8}$  & $\le 62.1$ \\ 
         & 30 & 48 & $\ge 1 - 10^{-9}$  & $\le 63.6$ \\ 
         & 34 & 47 & $\ge 1 - 10^{-10}$ & $\le 65.6$ \\ 
      \hline
    \end{tabular}
    $\:$
    \begin{tabular}{r|rr|l|r}
      \hline
               &        &     & Success     & Work \\
      $\Delta$ & $\tau$ & $t$ & probability & ($\log_2$) \\
      \thickhline
      100 &  4 & 52 & $\ge 0.9$          & $\le 55.6$ \\ 
          &  5 & 52 & $\ge 0.95$         & $\le 56.1$ \\ 
          &  7 & 52 & $\ge 0.99$         & $\le 57.1$ \\ 
          & 10 & 54 & $\ge 0.999$        & $\le 58.6$ \\ 
          & 14 & 52 & $\ge 1 - 10^{-4}$  & $\le 60.6$ \\ 
          & 17 & 52 & $\ge 1 - 10^{-5}$  & $\le 62.1$ \\ 
          & 20 & 54 & $\ge 1 - 10^{-6}$  & $\le 63.6$ \\ 
          & 24 & 52 & $\ge 1 - 10^{-7}$  & $\le 65.6$ \\ 
          & 27 & 52 & $\ge 1 - 10^{-8}$  & $\le 67.1$ \\ 
          & 30 & 53 & $\ge 1 - 10^{-9}$  & $\le 68.6$ \\ 
          & 34 & 52 & $\ge 1 - 10^{-10}$ & $\le 70.6$ \\ 
      \hline
      110 &  4 & 57 & $\ge 0.9$          & $\le 60.6$ \\ 
          &  5 & 57 & $\ge 0.95$         & $\le 61.1$ \\ 
          &  7 & 57 & $\ge 0.99$         & $\le 62.1$ \\ 
          & 10 & 59 & $\ge 0.999$        & $\le 63.6$ \\ 
          & 14 & 57 & $\ge 1 - 10^{-4}$  & $\le 65.6$ \\ 
          & 17 & 57 & $\ge 1 - 10^{-5}$  & $\le 67.1$ \\ 
          & 20 & 59 & $\ge 1 - 10^{-6}$  & $\le 68.6$ \\ 
          & 24 & 57 & $\ge 1 - 10^{-7}$  & $\le 70.6$ \\ 
          & 27 & 57 & $\ge 1 - 10^{-8}$  & $\le 72.1$ \\ 
          & 30 & 58 & $\ge 1 - 10^{-9}$  & $\le 73.6$ \\ 
          & 34 & 57 & $\ge 1 - 10^{-10}$ & $\le 75.6$ \\ 
      \hline
      120 &  4 & 62 & $\ge 0.9$          & $\le 65.6$ \\ 
          &  5 & 62 & $\ge 0.95$         & $\le 66.1$ \\ 
          &  7 & 62 & $\ge 0.99$         & $\le 67.1$ \\ 
          & 10 & 64 & $\ge 0.999$        & $\le 68.6$ \\ 
          & 14 & 62 & $\ge 1 - 10^{-4}$  & $\le 70.6$ \\ 
          & 17 & 62 & $\ge 1 - 10^{-5}$  & $\le 72.1$ \\ 
          & 20 & 64 & $\ge 1 - 10^{-6}$  & $\le 73.6$ \\ 
          & 24 & 62 & $\ge 1 - 10^{-7}$  & $\le 75.6$ \\ 
          & 27 & 62 & $\ge 1 - 10^{-8}$  & $\le 77.1$ \\ 
          & 30 & 63 & $\ge 1 - 10^{-9}$  & $\le 78.6$ \\ 
          & 34 & 62 & $\ge 1 - 10^{-10}$ & $\le 80.6$ \\ 
      \hline
      130 &  4 & 67 & $\ge 0.9$          & $\le 70.6$ \\ 
          &  5 & 67 & $\ge 0.95$         & $\le 71.1$ \\ 
          &  7 & 67 & $\ge 0.99$         & $\le 72.1$ \\ 
          & 10 & 69 & $\ge 0.999$        & $\le 73.6$ \\ 
          & 14 & 67 & $\ge 1 - 10^{-4}$  & $\le 75.6$ \\ 
          & 17 & 67 & $\ge 1 - 10^{-5}$  & $\le 77.1$ \\ 
          & 20 & 69 & $\ge 1 - 10^{-6}$  & $\le 78.6$ \\ 
          & 24 & 67 & $\ge 1 - 10^{-7}$  & $\le 80.6$ \\ 
          & 27 & 67 & $\ge 1 - 10^{-8}$  & $\le 82.1$ \\ 
          & 30 & 68 & $\ge 1 - 10^{-9}$  & $\le 83.6$ \\ 
          & 34 & 67 & $\ge 1 - 10^{-10}$ & $\le 85.6$ \\ 
      \hline
    \end{tabular}
  \end{center}
  \caption{The continuation of Tab.~\ref{tab:bounds} for larger values of~$\Delta$.
  See the caption of Tab.~\ref{tab:bounds} for details on how to read this table.
  Note that as $\Delta$ increases, so does the enumeration complexity and the associated memory requirements.
  At some point, the post-processing becomes infeasible to perform in practice.}
  \label{tab:bounds-continued}
\end{table}

\clearpage

\subsection{Supplementary tables for FF-DH}
\label{app:tables-ff-dh}
In this appendix, we tabulate the bounds in Thm.~\ref{thm:main} for a few select combinations of~$\tau$, $t$ and~$\Delta$, see Tab.~\ref{tab:bounds-ff-dh}, so as to illustrate how the advantage for FF-DH grows in~$\Delta$ compared to our earlier analysis in~\cite[App.~A.1, Tab.~2]{ekera-pp} where $\Delta = 0$.

\begin{table}[h]
  \begin{center}
    \begin{tabular}{r|rr|rrr|l|r|rr}
      \hline
           &     &     &          &        &     & Success            & Work       &      &     \\
      $l$  & $z$ & $m$ & $\Delta$ & $\tau$ & $t$ & probability        & ($\log_2$) & Ops  & Adv \\
      \thickhline
      2048 & 112 & 224 & 70 &  7 & 37 & $\ge 0.99$         & $\le 42.1$ &  532 &  7.6 \\
           &     &     & 50 & 10 & 29 & $\ge 0.999$        & $\le 33.6$ &  572 &  7.1 \\
           &     &     &  0 & 34 &  2 & $\ge 1 - 10^{-10}$ & $\le 22.1$ &  672 &  6.1 \\
      \hline
      3072 & 128 & 256 & 70 &  7 & 37 & $\ge 0.99$         & $\le 42.1$ &  628 &  9.7 \\
           &     &     & 50 & 10 & 29 & $\ge 0.999$        & $\le 33.6$ &  668 &  9.1 \\
           &     &     &  0 & 34 &  2 & $\ge 1 - 10^{-10}$ & $\le 22.1$ &  768 &  8.0 \\
      \hline
      4096 & 152 & 304 & 70 &  7 & 37 & $\ge 0.99$         & $\le 42.1$ &  772 & 10.5 \\
           &     &     & 50 & 10 & 29 & $\ge 0.999$        & $\le 33.6$ &  812 & 10.0 \\
           &     &     &  0 & 34 &  2 & $\ge 1 - 10^{-10}$ & $\le 22.1$ &  912 &  9.0 \\
      \hline
      6144 & 176 & 352 & 70 &  7 & 37 & $\ge 0.99$         & $\le 42.1$ &  916 & 13.3 \\
           &     &     & 50 & 10 & 29 & $\ge 0.999$        & $\le 33.6$ &  956 & 12.8 \\
           &     &     &  0 & 34 &  2 & $\ge 1 - 10^{-10}$ & $\le 22.1$ & 1056 & 11.6 \\
      \hline
      8192 & 200 & 400 & 70 &  7 & 37 & $\ge 0.99$         & $\le 42.1$ & 1060 & 15.4 \\
           &     &     & 50 & 10 & 29 & $\ge 0.999$        & $\le 33.6$ & 1100 & 14.8 \\
           &     &     &  0 & 34 &  2 & $\ge 1 - 10^{-10}$ & $\le 22.1$ & 1200 & 13.7 \\
      \hline
    \end{tabular}
  \end{center}
  \caption{The bounds in Thm.~\ref{thm:main} tabulated for $c = 1$ and a few select combinations of~$\Delta$, $\tau$ and~$t$ with respect to breaking FF-DH with an $m = 2z$-bit short exponent in a safe-prime group defined by an $l$-bit prime~$p$.
  Ekerå--Håstad's algorithm performs $o_{\text{EH}} = m + 2\ell = 3m - 2\Delta$ group operations quantumly, as reported in the column denoted ``Ops'', compared to $o_{\text{S}} = 2(l-1) - \Delta$ operations for Shor's original algorithm for the DLP~\cite{shor94, shor97} when modified to work in the large prime-order subgroup as in~\cite{ekera-revisiting} (with $\varsigma = 0$ and $\nu_{\ell} = -\Delta$).
  The advantage, defined as $o_{\text{S}} / o_{\text{EH}}$, is reported in the column denoted ``Adv'' rounded to the closest first decimal.
  The enumeration complexity is reported in the ``Work'' column.
  The bound in said column is also a bound on the enumeration complexity as given by $\log_2(\frac{4}{3} \sqrt{\pi \Nspace})$ for~$\Nspace$ as in Thm.~\ref{thm:main-dlp}.}
  \label{tab:bounds-ff-dh}
\end{table}

\noindent More specifically, we consider FF-DH in safe-prime groups with short exponents:

To introduce some notation, let~$p$ be an $l$-bit safe-prime --- i.e.~a prime such that $r = (p-1)/2$ is also prime --- and let $g \in \mathbb F_p^*$ be an element of order~$r$.
Then~$g$ generates an $r$-order subgroup~$\langle g \rangle$ of~$\mathbb F_p^*$.
Let $x = g^d$ for~$d$ an $m$-bit exponent.
Then our goal when breaking FF-DH is to compute the discrete logarithm $d = \log_g x$.

The best classical algorithms for computing discrete logarithms in~$\langle g \rangle$ for large~$p$ are the general number field sieve~(GNFS)~\cite{gnfs, gordon, schirokauer} that runs in time subexponential in~$l$, and generic algorithms such as Pollard's algorithms~\cite{pollard-rho-lambda, oorschot-wiener} that run in time $O(\sqrt{r})$ and $O(\sqrt{d})$.
For this reason, it is standard practice~\cite{nistsp800-56a, rfc3526, rfc7919} to use short $m = 2z$ bit exponents with FF-DH in safe-prime groups, for~$z$ the strength level provided by an $l$-bit prime with respect to attacks by the GNFS.
Selecting a significantly larger~$m$ would yield a significant performance penalty, but would not yield significantly better security with respect to the best classical attacks.

The~$z$ column in Tab.~\ref{tab:bounds-ff-dh} gives the strength level in bits according to the model used by NIST, see~\cite[Sect.~7.5]{fips-140-2-ig} and~\cite[App.~D, Tab.~25--26]{nistsp800-56a} for further details.
Note that there are other models in the literature.

\subsection{Supplementary tables for RSA}
\label{app:tables-rsa}
In this appendix, we tabulate the lower bound on the success probability, and the associated upper bound on the complexity, in Thm.~\ref{thm:main}, in $\tau$ and~$t$ for selected $\Delta$.

This when factoring large random RSA integers via the reduction from the RSA IFP to the short DLP, and when requiring that the lower bound on the success probability must be met when accounting for a reduction in the probability by a factor $f(\Delta)$ due to the generator~$g$ selected not having sufficiently large order.

We consider $\Delta = 20$ as in~\cite[App.~A.2.1]{ekera-pp}, and $\Delta \in \{9, 10, 13, 17, 21\}$ since these are the smallest~$\Delta$ that allow our prescribed lower bounds on the success probability to be met when accounting for the reduction factor, see Tab.~\ref{tab:bounds-rsa} below:

\begin{table}[h]
  \begin{center}
    \begin{tabular}{r|rr|l|l|r}
      \hline
               &        &     & Success     & Reduction  & Work \\
      $\Delta$ & $\tau$ & $t$ & probability & factor~$f(\Delta)$ & ($\log_2$) \\
      \thickhline
      20 &  4 & 12 & $\ge 0.9$          & $\ge 0.999867$  & $\le 15.6$ \\ 
         &  5 & 12 & $\ge 0.95$         &                 & $\le 16.1$ \\ 
         &  7 & 12 & $\ge 0.99$         &                 & $\le 17.1$ \\ 
         & 11 & 12 & $\ge 0.999$        &                 & $\le 19.1$ \\ 
      \hline
      \hline
       9 &  6 &  6 & $\ge 0.9$          & $\ge 0.9288$    & $\le 11.2$ \\ 
      \hline
      10 &  5 &  7 & $\ge 0.9$          & $\ge 0.95817$   & $\le 11.2$ \\ 
         &  7 &  8 & $\ge 0.95$         &                 & $\le 12.3$ \\ 
      \hline
      13 &  4 &  9 & $\ge 0.9$          & $\ge 0.99200$   & $\le 12.1$ \\ 
         &  5 &  9 & $\ge 0.95$         &                 & $\le 12.6$ \\ 
         &  9 & 10 & $\ge 0.99$         &                 & $\le 14.7$ \\ 
      \hline
      17 &  4 & 10 & $\ge 0.9$          & $\ge 0.999208$  & $\le 14.1$ \\ 
         &  5 & 10 & $\ge 0.95$         &                 & $\le 14.6$ \\ 
         &  7 & 11 & $\ge 0.99$         &                 & $\le 15.6$ \\ 
         & 13 & 10 & $\ge 0.999$        &                 & $\le 18.6$ \\ 
      \hline
      21 &  4 & 12 & $\ge 0.9$          & $\ge 0.9999278$ & $\le 16.1$ \\ 
         &  5 & 12 & $\ge 0.95$         &                 & $\le 16.6$ \\ 
         &  7 & 13 & $\ge 0.99$         &                 & $\le 17.6$ \\ 
         & 11 & 12 & $\ge 0.999$        &                 & $\le 19.6$ \\ 
         & 16 & 12 & $\ge 1 - 10^{-4}$  &                 & $\le 22.1$ \\ 
      \hline
    \end{tabular}
  \end{center}
  \caption{The lower bound on the success probability and associated upper bound on the complexity in Thm.~\ref{thm:main} tabulated in $\tau$ and~$t$ for $\Delta = 20$, and for $\Delta \in \{ 9, 10, 13, 17, 21 \}$, for $c = 1$.
  For a given lower bound on the success probability, the table gives~$t$ and~$\tau$ that minimize the enumeration complexity in group operations as given by $\log_2 \sqrt{\Nspace} + 3$ for~$\Nspace$ as in Thm.~\ref{thm:main}.
  This when requiring that the lower bound on the success probability must be met when accounting for a reduction in the probability by a factor $f(\Delta)$.
  The enumeration complexity is reported in the ``Work'' column.
  The bound in said column is also a bound on the enumeration complexity as given by $\log_2(\frac{4}{3} \sqrt{\pi \Nspace})$ for~$\Nspace$ as in Thm.~\ref{thm:main-dlp}.}
  \label{tab:bounds-rsa}
\end{table}

More specifically, for~$\Nmodulus = pq$ a large random RSA integer, the probability of~$g$ selected uniformly at random from~$\mathbb Z_{\Nmodulus}^*$ having order~$r \ge 2^{m+\ell} + (2^\ell - 1) d$ --- i.e.\ a sufficiently large order --- is lower bounded in~\cite[Lem.~4 in App.~A.2.2]{ekera-pp}, and asymptotically shown to be at least~$f(\Delta)$, see Tab.~\ref{tab:bounds-rsa} for concrete values.

Note that~$p$ and~$q$ are sampled from $[2, x]$ as $x \rightarrow \infty$ in~\cite[Lem.~4]{ekera-pp}, whereas~$p$ and~$q$ are $l$-bit primes in the analysis in~\cite[App.~A.2.2]{ekera-pp}.
As in~\cite{ekera-pp}, we assume that this distinction is not important.
We have verified the validity of this assumption through simulations, by sampling~$10^7$ random RSA integers $\Nmodulus = pq$, and exactly computing the order~$r$ of~$g$ selected uniformly random from~$\mathbb Z_{\Nmodulus}^*$ without explicitly computing~$g$ (see~\cite[Sect.~5.2.3]{ekera-phd-thesis} for a description of how to perform this computation).
Specifically, to sample $\Nmodulus = pq$, we first sample~$p$ and~$q$ independently and uniformly at random from the set of all $l$-bit primes for $l = 1024$.
We then return $\Nmodulus = pq$ if $p \neq q$ and~$pq$ is of length $2l = 2048$ bits, otherwise we try again.

As~$\Delta$ grows larger, so does the reduction factor~$f(\Delta)$.
However, the method used in~\cite[App.~A.2.2]{ekera-pp} to lower bound~$f(\Delta)$ is limited in that the computational complexity grows rapidly in~$\Delta$.
This explains why we do not include success probabilities ${\ge 1 - 10^{-5}}$ in Tab.~\ref{tab:bounds-rsa}.
One option for reaching greater success probabilities is to instead estimate the reduction factor via simulations.
For further details, see~\cite[Tab.~5.9]{ekera-phd-thesis}.

\clearpage

\section{Figures}
\label{app:figures}

In this appendix, we visualize the quantum circuits discussed in Sect.~\ref{sect:quantum-algorithm-implementation}.

{
\newcommand{\boxoffset}{1.30}
\newcommand{\boxwidth}{0.90}
\newcommand{\boxsep}{0.25}

{
\newcommand{\ctrlwidth}{11.05}
\newcommand{\workwidth}{11.05}

\begin{figure}[h!]
  \begin{center}
    \begin{tikzpicture}[scale=0.67, every node/.style={scale=0.67}]

      \foreach \y in {1.0}
      {
        \draw[densely dotted] ({\workwidth - 0.05}, {\y - 0.25}) -- ++ (0.20, 0) -- ++ (0, {2 + 2 * 0.25}) -- ++ (-0.20, 0);
        \draw[densely dotted] (0.05, {\y - 0.25}) -- ++ (-0.20, 0) -- ++ (0, {2 + 2 * 0.25}) -- ++ (0.20, 0);
      }

      \draw (-0.50, 2.0) node[rotate=90] {$\ket{g^0}$};
      \draw ({\workwidth + 0.45}, 2.0) node[rotate=90] {$\nu$ qubits};
      \draw ({\workwidth + 1.00}, 2.0) node[rotate=90] {$\ket{g^a \, x^{-b}}$};

      \draw[-stealth] (0, 3) -- ++ (\workwidth, 0);
      \draw[densely dotted, -stealth] (0, 2) -- ++ (\workwidth, 0);
      \draw[-stealth] (0, 1) -- ++ (\workwidth, 0);

      \foreach \y in {8}
      {
        \draw[-stealth] (0, \y) -- ++ (\ctrlwidth, 0);
        \draw[-stealth, densely dotted] (0, {\y + 1}) -- ++ (\ctrlwidth, 0);
        \draw[-stealth] (0, {\y + 2}) -- ++ (\ctrlwidth, 0);

        \draw[densely dotted] (0.05, {\y - 0.25}) -- ++ (-0.20, 0) -- ++ (0, {2 + 2 * 0.25}) -- ++ (0.20, 0);
        \draw[densely dotted] ({\ctrlwidth - 0.05}, {\y - 0.25}) -- ++ (0.20, 0) -- ++ (0, {2 + 2 * 0.25}) -- ++ (-0.20, 0);

        \draw (-0.50, {\y + 1}) node[rotate=90] {$\ket{0}$};
        \draw ({\ctrlwidth + 0.45}, {\y + 1}) node[rotate=90] {$m + \ell$ qubits};
        \draw ({\ctrlwidth + 1.00}, {\y + 1}) node[rotate=90] {yields~$j$};

        \foreach \t in {0, 2} {
          \draw[fill=white] (0.25, {\y + \t - 0.40}) rectangle ++ (0.80, 0.80) node[pos=0.5] {H};
        }
        \draw[fill=white, densely dotted] (0.25, {\y + 1 - 0.40}) rectangle ++ (0.80, 0.80) node[pos=0.5] {H};

        \draw[fill=white] ({\ctrlwidth - 2.20}, {\y - 0.40}) rectangle ++ (0.80, {2 + 2 * 0.40}) node[pos=0.5, rotate=90] {QFT};

        \foreach \t in {0, 2} {
          \draw[fill=white] ({\ctrlwidth - 1.15}, {\y + \t - 0.40}) rectangle ++ (0.80, 0.80);
        }
        \draw[fill=white, densely dotted] ({\ctrlwidth - 1.15}, {\y + 1 - 0.40}) rectangle ++ (0.80, 0.80);

        \foreach \t in {0, 1, 2} {
          \draw ({\ctrlwidth - 1.15 + 0.62}, {\y + \t - 0.05}) arc (30 : 150 : 0.25);
          \draw ({\ctrlwidth - 1.15 + 0.40}, {\y + \t - 0.09}) -- ++ (0.10, 0.20);
        }

        \begin{scope}
          \draw ({\boxoffset + \boxwidth / 2}, {\y + 2})
            node[circle, fill, inner sep=1pt] {} ++ (0, 0.05) node[above] {$\ket{a_0}$} ++ (0, -0.05)
              -- ++ (0, {3.40 - \y - 2});
          \draw[densely dotted] ({\boxoffset + \boxwidth + \boxsep + \boxwidth / 2}, {\y + 1})
            node[circle,fill,inner sep=1pt] {} ++ (0, 0.05) node[above] {$\ket{a_i}$} ++ (0, -0.05)
              -- ++ (0, {3.40 - \y - 1});
          \draw ({\boxoffset + 2 * (\boxwidth + \boxsep) + \boxwidth / 2}, {\y})
            node[circle,fill,inner sep=1pt] {} ++ (0, 0.05) node[above] {$\ket{a_{m + \ell - 1}}$} ++ (0, -0.05)
              -- ++ (0, {3.40 - \y});

          \draw[fill=white] (\boxoffset, {1 - 0.40}) rectangle ++ (\boxwidth, {2 + 2 * 0.40})
            node[pos=0.5, rotate=90] {$g^{2^0}$};
          \draw[densely dotted, fill=white] ({\boxoffset + \boxwidth + \boxsep}, {1 - 0.40})
            rectangle ++ (\boxwidth, {2 + 2 * 0.40})
              node[pos=0.5, rotate=90] {$g^{2^i}$};
          \draw[fill=white] ({\boxoffset + 2 * (\boxwidth + \boxsep)}, {1 - 0.40})
            rectangle ++ (\boxwidth, {2 + 2 * 0.40})
              node[pos=0.5, rotate=90] {$g^{2^{m + \ell - 1}}$};

          \draw[densely dotted]
            ({\boxoffset - 0.25}, 0.60)
              -- ({\boxoffset - 0.25}, 0.35)
                -- ({\boxoffset + 3 * \boxwidth + 2 * \boxsep + 0.25}, 0.35)
                  node[pos=0.5, below] {$m+\ell$ operations}
                    -- ({\boxoffset + 3 * \boxwidth + 2 * \boxsep + 0.25}, 0.60);
        \end{scope}
      }

      \foreach \y in {4.5}
      {
        \draw[-stealth] (0, \y) -- ++ (\ctrlwidth, 0);
        \draw[-stealth, densely dotted] (0, {\y + 1}) -- ++ (\ctrlwidth, 0);
        \draw[-stealth] (0, {\y + 2}) -- ++ (\ctrlwidth, 0);

        \draw[densely dotted] (0.05, {\y - 0.25}) -- ++ (-0.20, 0) -- ++ (0, {2 + 2 * 0.25}) -- ++ (0.20, 0);
        \draw[densely dotted] ({\ctrlwidth - 0.05}, {\y - 0.25}) -- ++ (0.20, 0) -- ++ (0, {2 + 2 * 0.25}) -- ++ (-0.20, 0);

        \draw (-0.50, {\y + 1}) node[rotate=90] {$\ket{0}$};
        \draw ({\ctrlwidth + 0.45}, {\y + 1}) node[rotate=90] {$\ell$ qubits};
        \draw ({\ctrlwidth + 1.00}, {\y + 1}) node[rotate=90] {yields~$k$};

        \foreach \t in {0, 2} {
          \draw[fill=white] (0.25, {\y + \t - 0.40}) rectangle ++ (0.80, 0.80) node[pos=0.5] {H};
        }
        \draw[fill=white, densely dotted] (0.25, {\y + 1 - 0.40}) rectangle ++ (0.80, 0.80) node[pos=0.5] {H};

        \draw[fill=white] ({\ctrlwidth - 2.20}, {\y - 0.40}) rectangle ++ (0.80, {2 + 2 * 0.40}) node[pos=0.5, rotate=90] {QFT};

        \foreach \t in {0, 2} {
          \draw[fill=white] ({\ctrlwidth - 1.15}, {\y + \t - 0.40}) rectangle ++ (0.80, 0.80);
        }
        \draw[fill=white, densely dotted] ({\ctrlwidth - 1.15}, {\y + 1 - 0.40}) rectangle ++ (0.80, 0.80);

        \foreach \t in {0, 1, 2} {
          \draw ({\ctrlwidth - 1.15 + 0.62}, {\y + \t - 0.05}) arc (30 : 150 : 0.25);
          \draw ({\ctrlwidth - 1.15 + 0.40}, {\y + \t - 0.09}) -- ++ (0.10, 0.20);
        }

        \begin{scope}[shift={(4, 0)}]
          \draw ({\boxoffset + \boxwidth / 2}, {\y + 2})
            node[circle, fill, inner sep=1pt] {} ++ (0, 0.05) node[above] {$\ket{b_0}$} ++ (0, -0.05)
              -- ++ (0, {3.40 - \y - 2});
          \draw[densely dotted] ({\boxoffset + \boxwidth + \boxsep + \boxwidth / 2}, {\y + 1})
            node[circle,fill,inner sep=1pt] {} ++ (0, 0.05) node[above] {$\ket{b_i}$} ++ (0, -0.05)
              -- ++ (0, {3.40 - \y - 1});
          \draw ({\boxoffset + 2 * (\boxwidth + \boxsep) + \boxwidth / 2}, {\y})
            node[circle,fill,inner sep=1pt] {} ++ (0, 0.05) node[above] {$\ket{b_{\ell - 1}}$} ++ (0, -0.05)
              -- ++ (0, {3.40 - \y});

          \draw[fill=white] (\boxoffset, {1 - 0.40}) rectangle ++ (\boxwidth, {2 + 2 * 0.40})
            node[pos=0.5, rotate=90] {$x^{-2^0}$};
          \draw[densely dotted, fill=white] ({\boxoffset + \boxwidth + \boxsep}, {1 - 0.40})
            rectangle ++ (\boxwidth, {2 + 2 * 0.40})
              node[pos=0.5, rotate=90] {$x^{-2^i}$};
          \draw[fill=white] ({\boxoffset + 2 * (\boxwidth + \boxsep)}, {1 - 0.40})
            rectangle ++ (\boxwidth, {2 + 2 * 0.40})
              node[pos=0.5, rotate=90] {$x^{-2^{\ell-1}}$};

          \draw[densely dotted]
            ({\boxoffset - 0.25}, 0.60)
              -- ({\boxoffset - 0.25}, 0.35)
                -- ({\boxoffset + 3 * \boxwidth + 2 * \boxsep + 0.25}, 0.35)
                  node[pos=0.5, below] {$\ell$ operations}
                    -- ({\boxoffset + 3 * \boxwidth + 2 * \boxsep + 0.25}, 0.60);
        \end{scope}
      }
    \end{tikzpicture}
  \end{center}

  \caption{A quantum circuit for inducing the state~\refeq{superposition} and measuring the two control registers yielding~$j$ and~$k$, respectively.
  In this figure, ${a = \sum_{i \, = \, 0}^{m+\ell-1} 2^i a_i}$ and ${b = \sum_{i \, = \, 0}^{\ell-1} 2^i b_i}$ where ${a_i, \, b_i \in \{0, 1\}}$, see Sect.~\ref{sect:quantum-algorithm-implementation}.
  The operations at the bottom are compositions under the group operation by classically pre-computed constant group elements.
  The bottom work register must be of sufficient length~$\nu$ to store a superposition of group elements and to perform the required group operations reversibly.}
  \label{fig:basic-circuit}
\end{figure}
}

{
\newcommand{\ctrlwidtha}{7.25}
\newcommand{\ctrlwidthb}{7.05}
\newcommand{\workwidth}{7.15 + 9.50}

\begin{figure}[h!]
  \begin{center}
    \begin{tikzpicture}[scale=0.67, every node/.style={scale=0.67}]

      \foreach \y in {1.0}
      {
        \draw[densely dotted] ({\workwidth - 0.05}, {\y - 0.25}) -- ++ (0.20, 0) -- ++ (0, {2 + 2 * 0.25}) -- ++ (-0.20, 0);
        \draw[densely dotted] (0.05, {\y - 0.25}) -- ++ (-0.20, 0) -- ++ (0, {2 + 2 * 0.25}) -- ++ (0.20, 0);
      }

      \draw (-0.50, 2.0) node[rotate=90] {$\ket{g^0}$};
      \draw ({\workwidth + 0.45}, 2.0) node[rotate=90] {$\nu$ qubits};
      \draw ({\workwidth + 1.00}, 2.0) node[rotate=90] {$\ket{g^a \, x^{-b}}$};

      \draw[-stealth] (0, 3) -- ++ (\workwidth, 0);
      \draw[densely dotted, -stealth] (0, 2) -- ++ (\workwidth, 0);
      \draw[-stealth] (0, 1) -- ++ (\workwidth, 0);

      \foreach \y in {4.5}
      {
        \draw[-stealth] (0, \y) -- ++ (\ctrlwidtha, 0);
        \draw[-stealth, densely dotted] (0, {\y + 1}) -- ++ (\ctrlwidtha, 0);
        \draw[-stealth] (0, {\y + 2}) -- ++ (\ctrlwidtha, 0);

        \draw[densely dotted] (0.05, {\y - 0.25}) -- ++ (-0.20, 0) -- ++ (0, {2 + 2 * 0.25}) -- ++ (0.20, 0);
        \draw[densely dotted] ({\ctrlwidtha - 0.05}, {\y - 0.25}) -- ++ (0.20, 0) -- ++ (0, {2 + 2 * 0.25}) -- ++ (-0.20, 0);

        \draw (-0.50, {\y + 1}) node[rotate=90] {$\ket{0}$};
        \draw ({\ctrlwidtha + 0.45}, {\y + 1}) node[rotate=90] {$m + \ell$ qubits};
        \draw ({\ctrlwidtha + 1.00}, {\y + 1}) node[rotate=90] {yields~$j$};

        \foreach \t in {0, 2} {
          \draw[fill=white] (0.25, {\y + \t - 0.40}) rectangle ++ (0.80, 0.80) node[pos=0.5] {H};
        }
        \draw[fill=white, densely dotted] (0.25, {\y + 1 - 0.40}) rectangle ++ (0.80, 0.80) node[pos=0.5] {H};

        \draw[fill=white] ({\ctrlwidtha - 2.20}, {\y - 0.40}) rectangle ++ (0.80, {2 + 2 * 0.40}) node[pos=0.5, rotate=90] {QFT};

        \foreach \t in {0, 2} {
          \draw[fill=white] ({\ctrlwidtha - 1.15}, {\y + \t - 0.40}) rectangle ++ (0.80, 0.80);
        }
        \draw[fill=white, densely dotted] ({\ctrlwidtha - 1.15}, {\y + 1 - 0.40}) rectangle ++ (0.80, 0.80);

        \foreach \t in {0, 1, 2} {
          \draw ({\ctrlwidtha - 1.15 + 0.62}, {\y + \t - 0.05}) arc (30 : 150 : 0.25);
          \draw ({\ctrlwidtha - 1.15 + 0.40}, {\y + \t - 0.09}) -- ++ (0.10, 0.20);
        }

        \begin{scope}
          \draw ({\boxoffset + \boxwidth / 2}, {\y + 2})
            node[circle, fill, inner sep=1pt] {} ++ (0, 0.05) node[above] {$\ket{a_0}$} ++ (0, -0.05)
              -- ++ (0, {3.40 - \y - 2});
          \draw[densely dotted] ({\boxoffset + \boxwidth + \boxsep + \boxwidth / 2}, {\y + 1})
            node[circle,fill,inner sep=1pt] {} ++ (0, 0.05) node[above] {$\ket{a_i}$} ++ (0, -0.05)
              -- ++ (0, {3.40 - \y - 1});
          \draw ({\boxoffset + 2 * (\boxwidth + \boxsep) + \boxwidth / 2}, {\y})
            node[circle,fill,inner sep=1pt] {} ++ (0, 0.05) node[above] {$\ket{a_{m + \ell - 1}}$} ++ (0, -0.05)
              -- ++ (0, {3.40 - \y});

          \draw[fill=white] (\boxoffset, {1 - 0.40}) rectangle ++ (\boxwidth, {2 + 2 * 0.40})
            node[pos=0.5, rotate=90] {$g^{2^0}$};
          \draw[densely dotted, fill=white] ({\boxoffset + \boxwidth + \boxsep}, {1 - 0.40})
            rectangle ++ (\boxwidth, {2 + 2 * 0.40})
              node[pos=0.5, rotate=90] {$g^{2^i}$};
          \draw[fill=white] ({\boxoffset + 2 * (\boxwidth + \boxsep)}, {1 - 0.40})
            rectangle ++ (\boxwidth, {2 + 2 * 0.40})
              node[pos=0.5, rotate=90] {$g^{2^{m + \ell - 1}}$};

          \draw[densely dotted]
            ({\boxoffset - 0.25}, 0.60)
              -- ({\boxoffset - 0.25}, 0.35)
                -- ({\boxoffset + 3 * \boxwidth + 2 * \boxsep + 0.25}, 0.35)
                  node[pos=0.5, below] {$m+\ell$ operations}
                    -- ({\boxoffset + 3 * \boxwidth + 2 * \boxsep + 0.25}, 0.60);
        \end{scope}
      }

      \begin{scope}[shift={(9.60, 0)}]
        \foreach \y in {4.5}
        {
          \draw[-stealth] (0, \y) -- ++ (\ctrlwidthb, 0);
          \draw[-stealth, densely dotted] (0, {\y + 1}) -- ++ (\ctrlwidthb, 0);
          \draw[-stealth] (0, {\y + 2}) -- ++ (\ctrlwidthb, 0);

          \draw[densely dotted] (0.05, {\y - 0.25}) -- ++ (-0.20, 0) -- ++ (0, {2 + 2 * 0.25}) -- ++ (0.20, 0);
          \draw[densely dotted] ({\ctrlwidthb - 0.05}, {\y - 0.25}) -- ++ (0.20, 0) -- ++ (0, {2 + 2 * 0.25}) -- ++ (-0.20, 0);

          \draw (-0.50, {\y + 1}) node[rotate=90] {$\ket{0}$};
          \draw ({\ctrlwidthb + 0.45}, {\y + 1}) node[rotate=90] {$\ell$ qubits};
          \draw ({\ctrlwidthb + 1.00}, {\y + 1}) node[rotate=90] {yields~$k$};

          \foreach \t in {0, 2} {
            \draw[fill=white] (0.25, {\y + \t - 0.40}) rectangle ++ (0.80, 0.80) node[pos=0.5] {H};
          }
          \draw[fill=white, densely dotted] (0.25, {\y + 1 - 0.40}) rectangle ++ (0.80, 0.80) node[pos=0.5] {H};

          \draw[fill=white] ({\ctrlwidthb - 2.20}, {\y - 0.40}) rectangle ++ (0.80, {2 + 2 * 0.40}) node[pos=0.5, rotate=90] {QFT};

          \foreach \t in {0, 2} {
            \draw[fill=white] ({\ctrlwidthb - 1.15}, {\y + \t - 0.40}) rectangle ++ (0.80, 0.80);
          }
          \draw[fill=white, densely dotted] ({\ctrlwidthb - 1.15}, {\y + 1 - 0.40}) rectangle ++ (0.80, 0.80);

          \foreach \t in {0, 1, 2} {
            \draw ({\ctrlwidthb - 1.15 + 0.62}, {\y + \t - 0.05}) arc (30 : 150 : 0.25);
            \draw ({\ctrlwidthb - 1.15 + 0.40}, {\y + \t - 0.09}) -- ++ (0.10, 0.20);
          }

          \begin{scope}
            \draw ({\boxoffset + \boxwidth / 2}, {\y + 2})
              node[circle, fill, inner sep=1pt] {} ++ (0, 0.05) node[above] {$\ket{b_0}$} ++ (0, -0.05)
                -- ++ (0, {3.40 - \y - 2});
            \draw[densely dotted] ({\boxoffset + \boxwidth + \boxsep + \boxwidth / 2}, {\y + 1})
              node[circle,fill,inner sep=1pt] {} ++ (0, 0.05) node[above] {$\ket{b_i}$} ++ (0, -0.05)
                -- ++ (0, {3.40 - \y - 1});
            \draw ({\boxoffset + 2 * (\boxwidth + \boxsep) + \boxwidth / 2}, {\y})
              node[circle,fill,inner sep=1pt] {} ++ (0, 0.05) node[above] {$\ket{b_{\ell - 1}}$} ++ (0, -0.05)
                -- ++ (0, {3.40 - \y});

            \draw[fill=white] (\boxoffset, {1 - 0.40}) rectangle ++ (\boxwidth, {2 + 2 * 0.40})
              node[pos=0.5, rotate=90] {$x^{-2^0}$};
            \draw[densely dotted, fill=white] ({\boxoffset + \boxwidth + \boxsep}, {1 - 0.40})
              rectangle ++ (\boxwidth, {2 + 2 * 0.40})
                node[pos=0.5, rotate=90] {$x^{-2^i}$};
            \draw[fill=white] ({\boxoffset + 2 * (\boxwidth + \boxsep)}, {1 - 0.40})
              rectangle ++ (\boxwidth, {2 + 2 * 0.40})
                node[pos=0.5, rotate=90] {$x^{-2^{\ell-1}}$};

            \draw[densely dotted]
              ({\boxoffset - 0.25}, 0.60)
                -- ({\boxoffset - 0.25}, 0.35)
                  -- ({\boxoffset + 3 * \boxwidth + 2 * \boxsep + 0.25}, 0.35)
                    node[pos=0.5, below] {$\ell$ operations}
                      -- ({\boxoffset + 3 * \boxwidth + 2 * \boxsep + 0.25}, 0.60);
          \end{scope}
        }
      \end{scope}
    \end{tikzpicture}
  \end{center}

  \caption{A quantum circuit, equivalent to that in Fig.~\ref{fig:basic-circuit}, for inducing the state~\refeq{superposition} and measuring the control registers yielding~$j$ and~$k$, respectively.
  Simply shifting the QFT and measurements left, and the initialization right, in the first and second control registers, respectively, in the circuit in Fig.~\ref{fig:basic-circuit}, yields this equivalent circuit.
  It first computes~$j$ and then computes~$k$ given~$j$.}
  \label{fig:j-then-k-circuit}
\end{figure}
}
}


\begin{thebibliography}{99}
  \bibitem{babai} L.\ Babai:
  On Lovász' lattice reduction and the nearest lattice point problem.
  Combinatorica 6(1) (1986), 1--13.

  \bibitem{nistsp800-56a} E.\ Barker et al.: 
  NIST SP 800-56A: Recommendation for Pair-Wise Key-Establishment Schemes Using Discrete Logarithm Cryptography, rev.~3 (2018).

  \bibitem{bloom} B.H.\ Bloom:
  Space/time trade-offs in hash coding with allowable errors.
  Comm.\ ACM 13(7) (1970), 422--426.

  \bibitem{cfs25} C.\ Chevignard, P.-A.\ Fouque and A.\ Schrottenloher.
  Reducing the number of qubits in quantum factoring.
  In: Crypto 2025. Lecture Notes in Computer Science (LNCS) 16001 (2025), 384--415.

  \bibitem{diffie-hellman} W.\ Diffie and M.E.\ Hellman:
  New Directions in Cryptography.
  IEEE Trans.\ Inf.\ Theory 22(6) (1976), 644--654.

  \bibitem{ekera-modifying} M.\ Ekerå:
  Modifying Shor's algorithm to compute short discrete logarithms.
  IACR ePrint Archive, Report 2016/1128 (2016).

  \bibitem{ekera-hastad} M.\ Ekerå and J.\ Håstad:
  Quantum algorithms for computing short discrete logarithms and factoring RSA integers.
  In: PQCrypto 2017. Lecture Notes in Computer Science (LNCS) 10346 (2017), 347--363.

  \bibitem{ekera-pp} M.\ Ekerå:
  On post-processing in the quantum algorithm for computing short discrete logarithms.
  Des.\ Codes Cryptogr.\ 88(11) (2020), 2313--2335.

  \bibitem{ekera-general} M.\ Ekerå:
  Quantum algorithms for computing general discrete logarithms and orders with tradeoffs.
  J.\ Math.\ Cryptol.\ 15(1) (2021), 359--407.

  \bibitem{ekera-success} M.\ Ekerå:
  On the success probability of quantum order finding.
  ACM Trans.\ Quantum Comput.\ 5(2):11 (2024), 1--40.

  \bibitem{ekera-revisiting} M.\ Ekerå:
  Revisiting Shor's quantum algorithm for computing general discrete logarithms.
  ArXiv 1905.09084v4 (2019--2024).

  \bibitem{ekera-phd-thesis} M.\ Ekerå:
  On factoring integers, and computing discrete logarithms and orders, quantumly.
  PhD thesis. KTH Royal Institute of Technology, Sweden (2024).

  \bibitem{quaspy} M.\ Ekerå:
  The Quaspy library for Python, v1.0.0a1 (2025).
  GitHub repository available at \url{https://github.com/ekera/quaspy}.

  \bibitem{galbraith-ruprai} S.\ Galbraith and R.S.\ Ruprai:
  An improvement to the Gaudry--Schost algorithm for multidimensional discrete logarithm problems.
  In: IMACC 2009. Lecture Notes in Computer Science (LNCS) 5921 (2009), 368--382.

  \bibitem{gaudry-schost} P.\ Gaudry and É.\ Schost:
  A low-memory parallel version of Matsuo, Chao, and Tsujii’s algorithm.
  In: ANTS 2004. Lecture Notes in Computer Science (LNCS) 3076 (2004), 208--222.

  \bibitem{gidney19} C.\ Gidney:
  Windowed quantum arithmetic.
  ArXiv 1905.07682v1 (2019).

  \bibitem{rfc7919} D.\ Gillmor:
  RFC 7919: Negotiated Finite Field Diffie-Hellman Ephemeral Parameters for Transport Layer Security (TLS) (2016).

  \bibitem{gordon} D.M.\ Gordon:
  Discrete logarithms in GF($p$) using the number field sieve.
  SIAM J.\ Discrete Math.\ 6(1) (1993), 124--138.

  \bibitem{griffiths-niu} R.B.\ Griffiths and C.-S.\ Niu:
  Semiclassical Fourier Transform for Quantum Computation.
  Phys.\ Rev.\ Lett.~76 (1996), 3228--3231.

  \bibitem{rfc3526} T.\ Kivinen and M.\ Kojo:
  RFC 3526: More Modular Exponentiation (MODP) Diffie-Hellman groups for Internet Key Exchange (2003).

  \bibitem{lagrange} J.-L.\ Lagrange:
  Recherches d'arithmétique, Œuvres complètes (tome 3),
  Nouveaux Mémoires de l'Académie royale des Sciences et Belles-Lettres de Berlin, années~1773 et~1775, (1773, 1775), 695--795. (Retrieved via Gallica.)

  \bibitem{lll} A.K.\ Lenstra, H.W.\ Lenstra, Jr.\ and L.\ Lovász:
  Factoring polynomials with rational coefficients.
  Math.\ Ann.~261 (1982), 515--534.

  \bibitem{gnfs} A.K.\ Lenstra, H.W.\ Lenstra, Jr., M.S.\ Manasse and J.M.\ Pollard:
  The number field sieve.
  In: Proceedings of the 22nd Annual ACM Symposium on Theory of Computing, STOC~'90 (1990), 564--572.

  \bibitem{ms22} A.\ May and L.\ Schlieper:
  Quantum period finding is compression robust.
  IACR Trans.\ Symmetric Cryptol., 2022(1) (2022), 183--211.

  \bibitem{vmi05} R.\ Van Meter and K.M.\ Itoh:
  Fast quantum modular exponentiation.
  Phys.\ Rev.~A 71(5):052320 (2005), 1--12.

  \bibitem{van-meter-phd-thesis} R.\ Van Meter:
  Architecture of a Quantum Multicomputer Optimized for Shor's Factoring Algorithm.
  PhD thesis. Keio University, Japan (2008).

  \bibitem{mosca-ekert} M.\ Mosca and A.\ Ekert:
  The Hidden Subgroup Problem and Eigenvalue Estimation on a Quantum Computer.
  In: QCQC 1998. Lecture Notes in Computer Science (LNCS) 1509 (1999), 174--188.

  \bibitem{nemes} G.\ Nemes:
  Error bounds for the asymptotic expansion of the Hurwitz zeta function.
  Proc.\ R.\ Soc.\ A.\ 473(2203):20170363 (2017), 1--16.

  \bibitem{nguyen} P.Q.\ Nguyen:
  Hermite's Constant and Lattice Algorithms.
  In: The LLL Algorithm: Survey and Applications (2010), 19--69, Springer Berlin Heidelberg.

  \bibitem{fips-140-2-ig} NIST and CCCS:
  Implementation Guidance for FIPS 140-2 and the Cryptographic Module Validation Program (2023).
  (Dated: March~17, 2023)

  \bibitem{oorschot-wiener} P.C.\ van Oorschot and M.J.\ Wiener:
  Parallel collision search with cryptanalytic applications.
  J.\ Cryptol. 12(1) (1999), 1--28.

  \bibitem{parker-plenio} S.\ Parker and M.B.\ Plenio:
  Efficient Factorization with a Single Pure Qubit and~$\log N$ Mixed Qubits.
  Phys.\ Rev.\ Lett.\ 85(14) (2000), 3049--3052.

  \bibitem{pollard-rho-lambda} J.M.\ Pollard:
  Monte Carlo Methods for Index Computation $(\text{mod } p)$.
  Math.\ Comput.\ 32(143) (1978), 918--924.

  \bibitem{rsa} R.L.\ Rivest, A.\ Shamir and L.\ Adleman:
  A Method for Obtaining Digital Signatures and Public-Key Cryptosystems.
  Commun.\ ACM 21(2) (1978), 120--126.

  \bibitem{schirokauer} O.\ Schirokauer:
  Discrete logarithms and local units.
  Phil.\ Trans.\ R.\ Soc.\ Lond.\ A 345(1676) (1993), 409--423.

  \bibitem{seifert} J.-P.\ Seifert:
  {Using Fewer Qubits in Shor's Factorization Algorithm via Simultaneous Diophantine Approximation.}
  In: CT-RSA 2001. Lecture Notes in Computer Science (LNCS) 2020 (2001), 319--327.

  \bibitem{shanks} D.\ Shanks:
  Class number, a theory of factorization, and genera.
  In: Proceedings of Symposia in Pure Mathematics, vol.~20 (1971), 415--440, American Mathematical Society.

  \bibitem{shor94} P.W.\ Shor:
  Algorithms for Quantum Computation: Discrete Logarithms and Factoring.
  In: Proceedings of the 35th Annual Symposium on Foundations of Computer Science, SFCS~'94 (1994), 124--134.

  \bibitem{shor97} P.W.\ Shor:
  Polynomial-time algorithms for prime factorization and discrete logarithms on a quantum computer.
  SIAM J.\ Comput.\ 26(5) (1997), 1484--1509.

  \bibitem{zalka} C.\ Zalka:
  Fast versions of Shor's quantum factoring algorithm.
  ArXiv quant-ph/9806084v1 (1998).
\end{thebibliography}
\end{document}